%% file: main-nufl.tex
\documentclass[11pt,letterpaper]{article}
\usepackage[lmargin=1.0in,rmargin=1.0in,bottom=1.0in,top=1.0in,twoside=False]{geometry}

\usepackage{fullpage,amssymb,amsmath}
\usepackage{graphicx}
\usepackage{enumerate}
\usepackage{tikz}
\usetikzlibrary{shapes}

\usepackage[T1]{fontenc}

 \usepackage{xcolor}
 \usepackage{mathtools}
\usepackage{microtype}
\usepackage{amsfonts}
\usepackage{comment}
\usepackage[english]{babel}
\usepackage{mathrsfs}
\usepackage[ruled,vlined]{algorithm2e}
\usepackage{blindtext}

\usepackage{trimspaces}
\usepackage{nccfoots}
\usepackage{setspace}
\usepackage{inconsolata}
\usepackage{libertine}
\usepackage[absolute]{textpos}

\usepackage{enumitem}
\usepackage{todonotes}

\usepackage{longtable}

\definecolor{blue}{rgb}{0.1,0.2,0.5}
\definecolor{brown}{rgb}{0.6,0.6,0.2}
\usepackage[ocgcolorlinks, linkcolor={blue}, citecolor={brown}]{hyperref}

\usepackage[amsmath,thmmarks,hyperref]{ntheorem}
\usepackage{cleveref}

\crefformat{page}{#2page~#1#3}%
\Crefformat{page}{#2Page~#1#3}%
\crefformat{equation}{#2(#1)#3}%
\Crefformat{equation}{#2(#1)#3}%
\crefformat{figure}{#2Figure~#1#3}%
\Crefformat{figure}{#2Figure~#1#3}%
\crefformat{section}{#2Section~#1#3}
\Crefformat{section}{#2Section~#1#3}
\crefformat{chapter}{#2Chapter~#1#3}
\Crefformat{chapter}{#2Chapter~#1#3}
\crefformat{chapter*}{#2Chapter~#1#3}
\Crefformat{chapter*}{#2Chapter~#1#3}
\crefformat{part}{#2Part~#1#3}
\Crefformat{part}{#2Part~#1#3}
\crefformat{enumi}{#2(#1)#3}
\Crefformat{enumi}{#2(#1)#3}

\usepackage{enumerate}

\usepackage{latexsym}

\theoremnumbering{arabic}
\theoremstyle{plain}
\theoremsymbol{}
\theorembodyfont{\itshape}
\theoremheaderfont{\normalfont\bfseries}
\theoremseparator{.}

\newtheorem{theorem}{Theorem}
\crefformat{theorem}{#2Theorem~#1#3}
\Crefformat{theorem}{#2Theorem~#1#3}

\newcommand{\newtheoremwithcrefformat}[2]{%
  \newtheorem{#1}[theorem]{#2}%
  \crefformat{#1}{##2\MakeUppercase#1~##1##3}%
  \Crefformat{#1}{##2\MakeUppercase#1~##1##3}%
}
\newcommand{\newseptheoremwithcrefformat}[2]{%
  \newtheorem{#1}{#2}%
  \crefformat{#1}{##2\MakeUppercase#1~##1##3}%
  \Crefformat{#1}{##2\MakeUppercase#1~##1##3}%
}

\newtheoremwithcrefformat{lemma}{Lemma}
\newtheoremwithcrefformat{proposition}{Proposition}
\newtheoremwithcrefformat{observation}{Observation}
\newtheoremwithcrefformat{conjecture}{Conjecture}
\newtheoremwithcrefformat{corollary}{Corollary}
\newseptheoremwithcrefformat{claim}{Claim}
\theorembodyfont{\upshape}
\newtheoremwithcrefformat{example}{Example}
\newtheoremwithcrefformat{remark}{Remark}
\newtheoremwithcrefformat{definition}{Definition}

\theoremstyle{nonumberplain}
\theoremheaderfont{\scshape}
\theorembodyfont{\normalfont}
\theoremsymbol{\ensuremath{\square}}
\newtheorem{proof}{Proof}

\theoremsymbol{\ensuremath{\lrcorner}}
\newtheorem{clproof}{Proof}

\def\cqedsymbol{\ifmmode$\lrcorner$\else{\unskip\nobreak\hfil
\penalty50\hskip1em\null\nobreak\hfil$\lrcorner$
\parfillskip=0pt\finalhyphendemerits=0\endgraf}\fi}

\newcommand{\Oh}{\mathcal{O}}

\newcommand{\Aa}{\mathcal{A}}

\newcommand{\Uu}{\mathcal{U}}
\newcommand{\Ww}{\mathcal{W}}

\newcommand{\eps}{\varepsilon}

\newcommand{\AlgName}{\mathtt{Solve}}

\newcommand{\N}{\mathbb{N}}
\newcommand{\R}{\mathbb{R}}

\renewcommand{\phi}{\varphi}
\renewcommand{\epsilon}{\varepsilon}

\newcommand{\dist}{\mathrm{dist}}

\newcommand{\prt}{\partial}

\newcommand{\openfac}{D}
\newcommand{\fac}{F}
\newcommand{\clients}{C}
\newcommand{\cost}{\mathsf{cost}}
\newcommand{\conncost}{\mathsf{conn}}
\newcommand{\opencost}{\mathsf{open}}

\renewcommand{\leq}{\leqslant}
\renewcommand{\geq}{\geqslant}

\newcommand{\layer}{\mathrm{layer}}
\newcommand{\OPT}{\mathsf{OPT}}

\newcommand{\scI}{\widetilde{I}}
\newcommand{\scSol}{\widetilde{D}}
\newcommand{\hintSol}{D^{\circ}}

\newcommand{\lr}{L}
\newcommand{\ring}{W}
\newcommand{\Z}{\mathbb{Z}}
\newcommand{\cluster}{\mathsf{cluster}}
\newcommand{\avgcost}{\mathsf{avgcost}}
\newcommand{\Cheap}{\scSol_{\mathrm{chp}}}
\newcommand{\Expensive}{\scSol_{\mathrm{exp}}}
\newcommand{\Far}{\mathsf{Far}}
\newcommand{\Close}{\mathsf{Close}}
\newcommand{\Offset}{\Psi}
\newcommand{\freefac}{D^{\mathrm{free}}}

\begin{document}
\title{A Polynomial-Time Approximation Scheme for Facility Location
  on Planar Graphs\thanks{This work is 
a part of projects CUTACOMBS (Ma. Pilipczuk) and TOTAL (Mi. Pilipczuk) that have received funding from the European Research Council (ERC) 
under the European Union's Horizon 2020 research and innovation programme (grant agreements No.~714704 and No.~677651, respectively).}}

\author{
Vincent Cohen-Addad\thanks{
 Sorbonne Universit\'e, CNRS, Laboratoire d'informatique de Paris 6, LIP6, F-75252 Paris, France \texttt{vincent.cohen.addad@ens-lyon.org}.
}
\and
Marcin~Pilipczuk\thanks{
  Institute of Informatics, University of Warsaw, Poland, \texttt{marcin.pilipczuk@mimuw.edu.pl}.
}
\and
Micha\l{}~Pilipczuk\thanks{
  Institute of Informatics, University of Warsaw, Poland, \texttt{michal.pilipczuk@mimuw.edu.pl}.
}
}

\begin{titlepage}
\def\thepage{}
\thispagestyle{empty}
\maketitle

\begin{textblock}{20}(0, 12.5)
\includegraphics[width=40px]{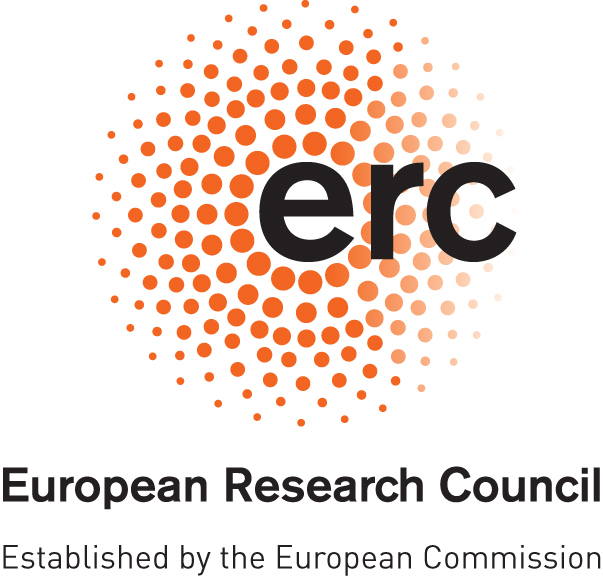}%
\end{textblock}
\begin{textblock}{20}(-0.25, 12.9)
\includegraphics[width=60px]{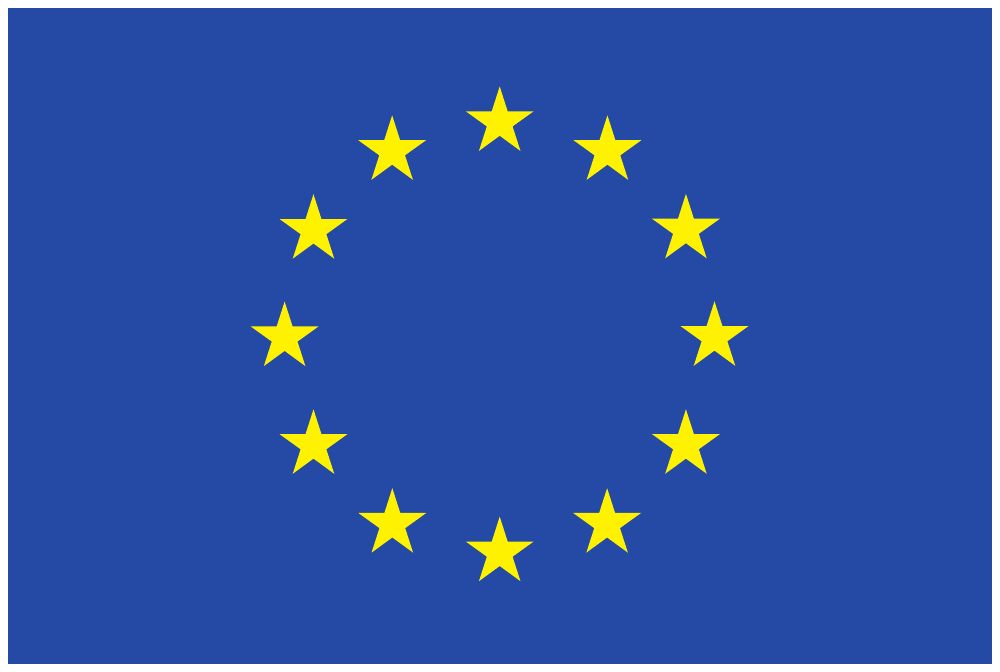}%
\end{textblock}

\input{abstract-nufl}

\end{titlepage}

\section{Introduction}
\input{intro-nufl}

\section{Reducing to the constant scope of the average costs}\label{sec:nufl-main}
\input{nufl-main}

\section{Dynamic programming algorithm}\label{sec:nufl-dp}
\input{nufl-dp}

\bibliographystyle{abbrv}
\bibliography{ref} 

\end{document}

%% file: abstract-nufl.tex
\begin{abstract}
  We consider the classic \textsc{Facility Location} problem
  on planar graphs (non-uniform, uncapacitated).
  Given an edge-weighted planar graph $G$, a set of clients
  $\clients\subseteq V(G)$, a set of facilities $\fac\subseteq V(G)$,
  and opening costs $\opencost \colon \fac \to \mathbb{R}_{\geq 0}$,
  the goal is to find a subset $\openfac$ of $\fac$ that minimizes
  $\sum_{c \in \clients} \min_{f \in \openfac} \dist(c,f) +
  \sum_{f \in \openfac} \opencost(f)$.  
  
  The  \textsc{Facility Location} problem remains one of the most classic
  and fundamental optimization problem for which it is not known whether
  it admits a {\em{polynomial-time approximation scheme}} (PTAS)
  on planar graphs despite significant effort for obtaining one.
  We solve this open problem by giving an algorithm that
  for any $\eps>0$, computes a
  solution of cost at most $(1+\eps)$ times the optimum
  in time $n^{2^{\Oh(\eps^{-2} \log (1/\eps))}}$.
\end{abstract}

%% file: intro-nufl.tex
We study the classic \textsc{Facility Location} objective in planar metrics.
Given an edge-weighted planar graph $G$, together 
with a set $\clients$ of vertices called \emph{clients},
a set $\fac$ of vertices called \emph{candidate facilities},
and \emph{opening costs} $\opencost \colon \fac \to \mathbb{R}_{\geq 0}$,
the \textsc{Facility Location} problem asks for
a subset $\openfac$ of $\fac$ that minimizes
$\sum_{c \in \clients} \min_{f \in \openfac} \dist(c,f) +
\sum_{f \in \openfac} \opencost(f)$.

The \textsc{Facility Location} objective is a model of choice when
trying to identify the best location for public infrastructures such as
hospitals, water tanks or fire stations, or when
looking for the best location for warehouses or delivery stores. More
recent applications also include prepositionning transportation
resources such as bikes, scooters, or cabs.
This has made \textsc{Facility Location}
a fundamental problem that attracted a lot
of attention over the years, both in theoretical computer science
and in operations research communities.
Since the problem is NP-hard, but one is often satisfied with a near-optimum solution, a large volume of work
was devoted to the design of approximation
algorithms~\cite{AGKMMP04,Hochbaum82,STA1997,JaV01},
culminating with the $1.488$-approximation algorithm by Li~\cite{Li13}.
Unfortunately, there is no chance of going much beyond this result, as the problem is known to be NP-hard to
approximate within factor better than $1.46$-approximation~\cite{GuK99}. 

Therefore, a natural route is to consider restricted metrics arising in
applications. For example,
when the underlying metric of the application is a road networks, 
the shortest path metric induced by edge-weighted planar graphs is
model of choice.
Thus, it has been a long standing open question whether
{\sc{Facility Locations}} admits a polynomial-time
approximation scheme on planar graphs. 
For the uniform case, this was resolved only recently in the
affirmative by Cohen-Addad et al.~\cite{local-search}
using a simple local search algorithm: given a current set of solution $\openfac$, determine whether
there exists a solution $\openfac'$ of better cost that differs from $\openfac$ by at most $\Oh(1/\eps^2)$ centers.
If so, take $\openfac'$ as the new solution and repeat,
otherwise output $\openfac$. 
However, no such approach is known to work in the nonuniform case and,
in fact, it is easy to show that the same local search heuristic would
provide a solution of cost at least twice the optimum in the worst-case
for planar inputs. This has been a major roadblock since local search
is the only technique we know so far for obtaining approximation schemes
to min-sum clustering objectives such as the classic
$k$-median, $k$-means or for uniform facility location, despite
a significant effort from the community.
In fact, and perhaps surprisingly,
such a situation is not unique. For the problem of computing a maximum
independent set of pseudo-disks, local search yields a PTAS in
the unweighted case and it remains an important open problem as
whether a PTAS exists for the weighted case~\cite{chan2012approximation}. 
Thus, obtaining a PTAS for the ``weighted'' version of some problems
seems a much harder task than for the unweighted case.

\medskip 

Our main result is a polynomial-time approximation
scheme for the (nonuniform, uncapacitated)
\textsc{Facility Location} problem in planar graphs. From a complexity
perspective, our result refutes APX-hardness of \textsc{Facility Location}
on planar graphs (unless $\mathsf{NP} = \mathsf{P}$). From a techniques
perspective, we believe that our approach provides a new set of
interesting tools, such as for example a ``metric-Baker'' layering
tailored to min-sum objectives (and so of a different nature
than the ``metric-Baker'' used for $k$-center in recent
works~\cite{FoxKS19,EisenstatKM14}). 
More formally, we show that following theorem.

\begin{theorem}\label{thm:nufl}
Given a {\sc{Facility Location}} instance $(G,\clients,\fac,\opencost)$, where
$G$ is a planar graph, and an accuracy parameter $\varepsilon > 0$,
one can in $n^{2^{\Oh(\eps^{-2} \log (1/\eps))}}$ time compute a solution
of cost at most $(1+\varepsilon)$ times the optimum cost.
\end{theorem}

We now describe the structure of the proof and our algorithm.
To do so conveniently, let us first introduce some terminology: we
define for a set $\openfac \subseteq \fac$, the \emph{connection cost} of
$\openfac$ is as $\conncost(\openfac, \clients) = \sum_{c \in \clients} \dist(c, \openfac)$ and the \emph{opening cost} of $\openfac$ as 
$\sum_{f \in \openfac} \opencost(f)$. 

The first step of our algorithm is to 
compute an $\Oh(1)$-approximate solution
to a modified input instance
where every opening cost is scaled down by a factor of $\eps$.
This solution  $\scSol$ is computed through a greedy procedure and has
the following satisfying properties: it is still an
$\Oh(\eps^{-1})$-approximation to the original instance, and
interestingly it reveals
a lot of structure of the input graph metric. This structure will
be crucial for the 
proof of Theorem~\ref{thm:nufl}. Indeed, the proof of the theorem and
our algorithm can be broken into two pieces. The first one consists in
a partitioning of the instance into separate, more structured,
and almost independent sub-instances (based on the output of the
greedy procedure). The second one is a heavily
technical dynamic programming algorithm for solving these sub-instances.

To understand how the two pieces articulate, we need to 
introduce a couple of definitions.
Let $f \in \scSol$ be an opened facility and let $\cluster(f)$ be the set of clients connected to $f$ in the solution $\scSol$
(i.e., $\cluster(f)$ consists of these clients $c \in \clients$ for which $f$ is the closest facility from $\scSol$). 
The \emph{average cost} of $\cluster(f)$ is defined as:
$$\avgcost(f) = \frac{\opencost(f)+\sum_{c\in \cluster(f)} \dist(c,f)}{|\cluster(f)|}.$$

At a high-level,
the sub-instances will be defined by dividing the metric space according
to the clustering induced by $\scSol$: putting in the same
instance the clusters of $\scSol$ that have roughly the same $\avgcost$
values. More concretely, a
deep analysis of the structure of the approximate solution $\scSol$
and an intricate Baker-type layering step based on
average costs of the facilities of $\scSol$ yields an instance
such that (i) all values of $\avgcost(f)$ for $f \in \scSol$ are within constant ratio from each other, and (ii) for every $f \in \scSol$
and $c \in \cluster(f)$ the distance $\dist(c,f)$ is within constant ratio of $\avgcost(f)$. 
This is described in Section~\ref{sec:nufl-main}.
The second part of the algorithm described in Section~\ref{sec:nufl-dp},
consists mainly of our technical dynamic programming algorithm for solving
the instances produced in the first part.


%% file: nufl-main.tex
\paragraph*{Setup.}
We shall work with an instance $I=(G,\clients,\fac,\opencost)$ where $G$ is a planar edge-weighted graph,
$\clients \subseteq V(G)$ is a set of clients,
$\fac \subseteq V(G)$ is a set of facilities,
and $\opencost \colon \fac \to \mathbb{R}_{\geq 0}$ defines the opening cost of facilities.
We shall assume that $G$ is embedded in a sphere and that distances between pairs of vertices of $G$ are finite and pairwise distinct.

For a set of clients $S\subseteq \clients$ and solution $R\subseteq \fac$, by $\conncost(S,R)$ we denote the contribution of clients from $S$ to the connection cost of $R$ and by $\opencost(R)$ the opening cost of $R$. That is,
$$\conncost(S,R)=\sum_{c\in S} \min_{f\in R} \dist(c,f)\quad\mathrm{and}\quad
\opencost(R)=\sum_{f\in R}\opencost(f).$$
We write $\conncost(R)$ for $\conncost(\clients,R)$.
Thus, the cost of $R$ is defined as $\cost(R)=\conncost(R)+\opencost(R)$.
For the remainder of this section, let us fix some optimum solution $\openfac$ in $I$, and we denote $\OPT=\cost(\openfac)$.

We consider the accuracy parameter $\eps>0$; w.l.o.g. we assume that $\eps<1/10$. Our goal is to compute a $(1+c\eps)$-approximate solution for some constant $c$, so that $\eps$ can be scaled appropriately at the end.

Recall that the considered problem admits a constant-factor approximation for the problem: as shown by Li~\cite{Li13}, given an instance of non-uniform facility location one can
in polynomial time find a solution of cost at most
$\alpha$ times the optimum, where $\alpha=1.488$. We apply this algorithm to the input instance, obtaining a solution $D' \subseteq \fac$, and we rescale the distances and the opening costs by the same factor so that 
$$\cost(D') = \eps^{-1}\cdot (|\fac| + |\clients|\cdot |E(G)|).$$
Note that this means that the total contribution of edges of length less than $1$ and facilities of opening cost less than $1$ to any solution is bounded by
$|\fac| + |\clients|\cdot |E(G)|\leq \eps\cdot \cost(D')\leq \alpha\eps\cdot \OPT$,
where by $\OPT$ we denote the optimum cost of a solution.
Thus, at the cost of paying at most $\eps\cdot \cost(D') \leq \alpha\eps \cdot \OPT$
we may assume that all edges of length less than~$1$ can be traversed for free, hence we may simply contract them.
Similarly, we zero the opening costs of all facilities whose opening cost is less than~$1$.
Therefore, we assume that all edges in $G$ have weight at least~$1$ and all opening costs are either $0$ or at least~$1$, while
\begin{equation}\label{eq:OPTorder}
\OPT = \Theta(\eps^{-1} \cdot (|\fac| + |\clients| \cdot |E(G)|)).
\end{equation}

\paragraph*{Robust approximate solution.}
Let us consider the modified instance
$$\scI = (G,\clients,\fac,\eps\cdot \opencost);$$
that is, the instance is the same as $I$ but all the opening costs are scaled down by a multiplicative factor of~$\eps$.
For a solution $R\subseteq \fac$, we denote the cost of $R$ in the instance $\scI$ by $\cost(R;\scI)$; note that
$\cost(R;\scI)=\conncost(R)+\eps\cdot \opencost(R)$.
Note that for any $R\subseteq \fac$, we have $\eps\cdot \cost(R)\leq \cost(R;\scI)\leq \cost(R)$.

We apply the aforementioned $\alpha$-approximation algorithm of Li~\cite{Li13} to the instance $\scI$.
Furthermore, we will need the following property from the returned approximate solution $\scSol$:
\begin{eqnarray}
\cost(\scSol\cup \{f\};\scI)\geq \cost(\scSol;\scI) & \qquad & \textrm{for every }f\in \fac;\label{eq:local-add}
\end{eqnarray}
This is trivially true for any $f \in \scSol$ and to ensure that this holds for every $f$,  we make use of the following greedy process.
As long as there exists a facility $f \in \fac \setminus \scSol$ that violates the condition above, we add it to $\scSol$.

Finally, at the end of this greedy process we remove from $\scSol$ all facilities that do not serve any client, that is, we remove all facilities $f \in \scSol$ such that
for every $c \in \clients$ we have $\dist(c, \scSol) < \dist(c, f)$. Note that this step does not increase the cost of $\scSol$ and does not break
property~\eqref{eq:local-add}. 
We now start analysing the structure of $\scSol$.

We start by verifying that $\scSol$ is actually an $\Oh(\eps^{-1})$-approximate solution in the original instance.

\begin{lemma}\label{lem:scSol-apx}
We have $\cost(\scSol)\leq \alpha\eps^{-1}\cdot \OPT$.
\end{lemma}
\begin{proof}
Recalling that $\openfac$ is an optimum solution in $I$, we have that
$$\cost(\openfac;\scI)\leq \cost(\openfac)=\OPT.$$
On the other hand, $\scSol$ is an $\alpha$-approximate solution in $\scI$, hence
$$\cost(\scSol;\scI)\leq \alpha\cdot \cost(\openfac;\scI)$$
Finally, as observed before we have
$$\eps \cdot \cost(\scSol)\leq \cost(\scSol;\scI).$$
Combining the above three inequalities yields the claim.
\end{proof}

Let $R\subseteq \fac$ be a nonempty set of facilities.
For a facility $f\in R$, the {\em{$R$-cluster}} of $f$, denoted $\cluster(f,R)$, is the set of all clients that are served by $f$ in the solution $R$; that is:
$$\cluster(f,R)=\{c\in \clients\colon \dist(c,f)=\min_{g\in R} \dist(c,g)\}.$$
Note that since distances between pairs of vertices in $G$ are pairwise different, the $R$-clusters are pairwise disjoint.
In the sequel we will most often work with $\scSol$-clusters, hence we use shorthands: a {\em{cluster}} means a $\scSol$-cluster and for $f\in \scSol$ we denote $\cluster(f)=\cluster(f,\scSol)$.

The next lemma intuitively says the following: for any subset of clients, its connection cost in $\scSol$ is not much larger than its connection cost $\openfac$.

\begin{lemma}\label{lem:reconnect}
For any subset of clients $S\subseteq \clients$ we have
$$\conncost(S,\scSol)\leq \conncost(S,\openfac)+\eps\cdot \opencost(\openfac).$$
\end{lemma}
\begin{proof}
For any $f\in \openfac$, let $$\sigma(f)=\conncost(\cluster(f,\openfac)\cap S,D)+\eps\cdot \opencost(f).$$
Observe that the right hand side of the inequality is equal to $\sum_{f\in \openfac} \sigma(f)$.

Consider modifying the solution $\scSol$ by opening facility $f$, for any $f\in \openfac$, and applying~\eqref{eq:local-add}.
If in solution $\scSol\cup \{f\}$ we consider directing all clients of $\cluster(f,\openfac)\cap S$ to $f$ and all other clients as in $\scSol$, then
\begin{eqnarray*}
0 & \leq & \cost(\scSol\cup \{f\};\scI)-\cost(\scSol;\scI)\\
  & \leq & \conncost(\cluster(f,\openfac)\cap S,\openfac)-\conncost(\cluster(f,\openfac)\cap S,\scSol)+\eps\cdot \opencost(f)\\
  & =    & \sigma(f)-\conncost(\cluster(f,\openfac)\cap S,\scSol).
\end{eqnarray*}
By summing the above inequality through all $f\in \openfac$, we infer that
$$0\leq \sum_{f\in \openfac} \sigma(f) - \sum_{f\in \openfac} \conncost(\cluster(f,\openfac)\cap S,\scSol) = \left(\conncost(S,\openfac)+\eps\cdot \opencost(\openfac)\right) - \conncost(S,\scSol).$$
This establishes the claim.
\end{proof}

For any $f\in \scSol$, we define the {\em{average cost}} of $f$ as
$$\avgcost(f) = \frac{\opencost(f)+\sum_{c\in \cluster(f)} \dist(c,f)}{|\cluster(f)|}.$$
Recall here that $\cluster(f)$ is nonempty for each $f\in \scSol$ as we removed from $\scSol$ all facilites that do not serve any clients.
Moreover, we have 
\begin{equation}\label{eq:avgcost-sum}
\cost(\scSol) = \sum_{f\in \scSol} |\cluster(f)|\cdot \avgcost(f).
\end{equation}

Next, we prove that for every cluster $\cluster(f)$ for any $f\in \scSol$, there is always a facility of the optimum solution $\openfac$ that is not far from $f$, measured in terms of $\avgcost(f)$.
We first state the lemma in a very abstract form so that we can apply it later in various settings.

\begin{lemma}\label{lem:close-opt-abs}
Let $I = (G,\clients,\fac,\opencost)$ be a \textsc{Non-Uniform Facility Location} instance, $\openfac\subseteq \fac$ a nonempty set of facilities, $K\subseteq \clients$ a nonempty set of clients, 
and let $f \notin D$ be a facility. Assume that 
$$\dist(f, \openfac) > \frac{2}{|K|}\cdot  \left (\opencost(f) + \sum_{c \in K} \dist(c,f)\right).$$
Then
$\cost(\openfac; I) > \cost(\openfac \cup \{f\} ; I)$.
\end{lemma}
\begin{proof}
Let
$$a := \frac{\opencost(f) + \sum_{c \in K} \dist(c,f)}{|K|}.$$
Observe that every client $c\in \cluster(f)$ has to be served in solution $\openfac$ by a facility that is at distance more than $2a$ from $f$, implying by triangle inequality that
$$\min_{g\in \openfac} \dist(c,g)>2a - \dist(c,f).$$
Take solution $\openfac\cup \{f\}$. By considering directing all the clients of $K$ to $f$, and all other clients as in $\openfac$, we observe that
\begin{eqnarray*}
\cost(\openfac\cup \{f\})-\cost(\openfac) & \leq & \sum_{c\in K} \dist(c,f) - \sum_{c\in K} \min_{g\in \openfac} \dist(c,g) + \opencost(f) \\
                                          & <    & \left(2\sum_{c\in K} \dist(c,f) + \opencost(f)\right) - 2|K|\cdot a \\
                                          & \leq    & 2|K|\cdot a - 2|K|\cdot a = 0.
\end{eqnarray*}
This implies that $\cost(\openfac\cup \{f\})<\cost(\openfac)$ as desired.
\end{proof}

\begin{corollary}\label{cor:close-opt}
For every $f\in \scSol$ there exists $g\in \openfac$ such that $\dist(f,g)\leq 2\cdot \avgcost(f)$.
\end{corollary}
\begin{proof}
The claim is obvious for $f \in \openfac$. Otherwise, we 
apply Lemma~\ref{lem:close-opt-abs} to the instance $I$, optimum solution $\openfac$, the facility $f$, and $K = \cluster(f)$.
The optimality of $\openfac$ implies then that $\dist(f, \openfac) \leq 2\cdot \avgcost(f)$.
\end{proof}

\paragraph*{Concentrating the clusters.}
We now analyze every cluster $\cluster(f)$ for $f\in \scSol$ and show that, at the cost of changing the value of $\OPT$ only slightly, 
we may assume that all clients of $\cluster(f)$ have connection cost w.r.t. $\scSol$ not differing much from $\avgcost(f)$. 
More precisely, we would like to get rid of clients that are far and close according to the following definition: for $f\in \scSol$, let
\begin{eqnarray*}
\Far(f) & = & \{c\in \cluster(f)\colon \dist(c,f)>\eps^{-2}\cdot \avgcost(f)\},\\
\Close(f) & = & \{c\in \cluster(f)\colon \dist(c,f)<\eps^2\cdot \avgcost(f)\}.
\end{eqnarray*}
Moreover, we define
$$\Far = \bigcup_{f\in \scSol} \Far(f)\qquad\textrm{and}\qquad \Close=\bigcup_{f\in \scSol}\Close(f).$$

Let
$$\Offset = \conncost(\Far,\scSol).$$
For each $f\in \scSol$ let us pick any vertex $x(f)$ of $G$ that is at distance exactly $\eps^2\cdot \avgcost(f)$ from $f$ (subdividing some edge, if a priori there is none).
Construct $\clients'$ from $\clients$ by performing the following operation for each $f\in \scSol$: move all clients of $\Far(f)\cup \Close(f)$ to $x(f)$, thus placing $|\Far(f)|+|\Close(f)|$ clients at $x(f)$.
Similarly, for $f\in \scSol$ we define $\cluster'(f)$ to be the image of $\cluster(f)$ under this operation, i.e. with clients from $\Far(f)\cup \Close(f)$ replaced as above.

Let 
$$I'=(G,\clients',\fac,\opencost);$$
that is, $I'$ is constructed from $I$ by replacing the client set with $\clients'$. 
Let $\OPT'$ be the minimum cost of a solution in the instance $I'$.
We now verify that in order to find near-optimum solution to $I$, it suffices to find a near-optimum solution to $I'$.

\begin{lemma}\label{lem:moving}
We have
$$\OPT'\leq (1+6\alpha\eps)\OPT - \Offset$$
Moreover, for every $R\subseteq \fac$, we have
$$\cost(R;I)\leq \cost(R;I') + \Offset + 3\alpha\eps\cdot \OPT.$$
\end{lemma}
\begin{proof}
For the first inequality, note that we have 
\begin{eqnarray}
\conncost(\clients',\openfac) & =    & \conncost(\clients,\openfac) + \sum_{f\in \scSol}\,\sum_{c\in \Far(f)} (\dist(x(f),\openfac)-\dist(c,\openfac))\nonumber \\
                              &      & + \sum_{f\in \scSol}\,\sum_{c\in \Close(f)} (\dist(x(f),\openfac)-\dist(c,\openfac)) \nonumber \\
                              & \leq & \conncost(\clients,\openfac) + \sum_{f\in \scSol}\,\sum_{c\in \Far(f)} (\dist(x(f),\openfac)-\dist(c,\openfac))\nonumber \\
                              &      & + \sum_{f\in \scSol}\,\sum_{c\in \Close(f)} \dist(c,x(f)).\label{eq:ccp}
\end{eqnarray}
Let us analyze the last summand first. Observe that for each $f\in \scSol$ and $c\in \Close(f)$, we have 
$$\dist(c,x(f))\leq \dist(c,f)+\dist(f,x(f))\leq 2\eps^2\cdot \avgcost(f).$$
Thus, using~\eqref{eq:avgcost-sum} we have
\begin{eqnarray}
\sum_{f\in \scSol}\,\sum_{c\in \Close(f)} \dist(c,x(f)) & \leq & \sum_{f\in \scSol} |\Close(f)|\cdot 2\eps^2\cdot \avgcost(f) \nonumber \\
                                                      & \leq & 2\eps^2\cdot \sum_{f\in \scSol} |\cluster(f)|\cdot \avgcost(f) \nonumber \\
                                                      & =    & 2\eps^2\cdot \cost(\scSol) \leq 2\alpha\eps\cdot \OPT.\label{eq:close}
\end{eqnarray}
We are left with analyzing the middle summand of the right hand side of~\eqref{eq:ccp}.
Observe that we have
\begin{equation*}
\sum_{f\in \scSol}\,\sum_{c\in \Far(f)} \dist(c,\openfac) = \conncost(\Far,\openfac).
\end{equation*}
By Lemma~\ref{lem:reconnect} applied to $S=\Far$, we infer that
\begin{equation*}
\Offset=\conncost(\Far,\scSol)\leq \conncost(\Far,\openfac)+\eps\cdot \OPT,
\end{equation*}
and thus we have
\begin{equation}\label{eq:far1}
\sum_{f\in \scSol}\,\sum_{c\in \Far(f)} \dist(c,\openfac) \geq \Offset-\eps\cdot \OPT.
\end{equation}
For every $f\in \scSol$, let $g(f)$ be the facility of $\openfac$ that is closest to $f$. By Corollary~\ref{cor:close-opt} we have that
$$\dist(f,g(f))\leq 2\cdot \avgcost(f).$$
Now, for every $c\in \Far(f)$ we have
\begin{eqnarray*}
3\cdot \dist(c,f) & \geq & 3\eps^{-2}\cdot  \avgcost(f) \\
                  & \geq & \eps^{-2}\cdot\dist(f,g(f))+\eps^{-4}\cdot \dist(f,x(f))\\
                  & \geq & \eps^{-2}\cdot \dist(x(f),g(f)) \geq \eps^{-2}\cdot \dist(x(f),\openfac),
\end{eqnarray*}
where in the second step we used $\dist(f,x(f))=\eps^2 \cdot \avgcost(f)$.
Summing this inequality through all $f\in \scSol$ and $c\in \Far(f)$ we obtain that
$$\conncost(\Far,\scSol)=\sum_{f\in \scSol}\sum_{c\in \Far(f)} \dist(c,f) \geq \frac{\eps^{-2}}{3} \sum_{f\in \scSol}\sum_{c\in \Far(f)} \dist(x(f),\openfac),$$
which means that 
\begin{equation}\label{eq:far11}
\sum_{f\in \scSol}\sum_{c\in \Far(f)} \dist(x(f),\openfac)\leq 3\eps^2\cdot \conncost(\Far,\scSol)\leq 3\eps^2\cdot \cost(\scSol)\leq 3\alpha\eps \cdot \OPT.
\end{equation}
By combining~\eqref{eq:ccp},~\eqref{eq:close},~\eqref{eq:far1}, and~\eqref{eq:far11} we infer that
\begin{eqnarray*}
\OPT' & \leq & \cost(\openfac;I') \\
      & =    & \opencost(\openfac) + \conncost(\clients',\openfac) \\
      & \leq & \opencost(\openfac) + \conncost(\clients,\openfac) + 2\alpha\eps\cdot \OPT - \Offset + \eps\cdot \OPT + 3\alpha\eps\cdot \OPT \\
      & \leq & \cost(\openfac;I) + 6\alpha\eps\cdot \OPT -\Offset = (1+6\alpha\eps)\OPT-\Offset.
\end{eqnarray*}
This establishes the first inequality.

For the second inequality, again by triangle inequality we have
\begin{equation}\label{eq:cpc}
\conncost(\clients,R) \leq \conncost(\clients',R) + \sum_{f\in \scSol}\sum_{c\in \Far(f)} \dist(c,x(f)) + \sum_{f\in \scSol}\sum_{c\in \Close(f)} \dist(c,x(f)).
\end{equation}
The last summand has already been estimated in~\eqref{eq:close}, so we are left with analyzing the middle summand. 
Observe that for each $f\in \scSol$ and $c\in \Far(f)$, we have
$$\dist(c,x(f))\leq \dist(c,f)+\dist(f,x(f))\leq (1+\eps^4)\cdot \dist(c,f),$$
where the last inequality follows from $\dist(c,f)\geq \eps^{-2}\cdot \avgcost(f)$ and $\dist(f,x(f))=\eps^{2}\cdot \avgcost(f)$.
Thus, we have
\begin{eqnarray}
\sum_{f\in \scSol}\sum_{c\in \Far(f)} \dist(c,x(f)) & \leq & (1+\eps^4)\cdot \sum_{f\in \scSol}\sum_{c\in \Far(f)} \dist(c,f)\nonumber \\
                                                    & =    & (1+\eps^4)\cdot \conncost(\Far,\scSol)\nonumber\\
                                                    & =    & \Offset + \eps^4\cdot \conncost(\Far,\scSol)\nonumber\\
                                                    & \leq & \Offset + \eps^4\cdot \cost(\scSol) \leq \Offset + \alpha\eps^3\cdot \OPT.\label{eq:far2}
\end{eqnarray}

By combining~\eqref{eq:cpc},~\eqref{eq:close}, and~\eqref{eq:far2} we obtain that
\begin{eqnarray*}
\cost(R;I) & =    & \opencost(R)+\conncost(\clients,R) \\
           & \leq & \opencost(R)+\conncost(\clients',R) + \Offset + 3\alpha\eps\cdot \OPT \\
           & =    & \cost(R;I') + \Offset + 3\alpha\eps\cdot \OPT.
\end{eqnarray*}
This concludes the proof.
\end{proof}

Lemma~\ref{lem:moving} immediately implies the following: any near-optimum solution to $I'$ is also a near-optimum solution to $I$.

\begin{corollary}\label{cor:move-equiv}
For any $R\subseteq \fac$, if 
$$\cost(R;I')\leq (1+\gamma)\OPT'+\delta,$$
for some $\gamma,\delta\geq 0$, then
$$\cost(R;I)\leq (1+2\gamma+8\alpha\eps)\OPT+\delta.$$
\end{corollary}
\begin{proof}
First, note that $\OPT'\leq (1+5\alpha\eps)\OPT-\Offset\leq 2\cdot\OPT$.
Then we have
\begin{eqnarray*}
\cost(R;I) & \leq & \cost(R;I') + \Offset + 3\alpha\eps\cdot \OPT \\
           & \leq & (1+\gamma)\OPT' + \delta + \Offset + 3\alpha\eps\cdot \OPT \\
           & \leq & \OPT' + 2\gamma\OPT +\delta + \Offset + 3\alpha\eps\cdot \OPT \\
           & \leq & (1+5\alpha\eps)\OPT-\Offset + 2\gamma\OPT +\delta + \Offset +3\alpha\eps\cdot \OPT = (1+2\gamma + 8\alpha\eps)\OPT+\delta,
\end{eqnarray*}
as claimed.
\end{proof}

Thus, by Corollary~\ref{cor:move-equiv} we may focus on finding a near-optimum solution to instance $I'$ instead of $I$.
The instance $I'$, however, has the following concentration property that will be useful later on: for every $f\in \scSol$ and $c\in \cluster'(f)$, we have
$$\eps^{2}\cdot \avgcost(f)\leq \dist(c,f)\leq \eps^{-2}\cdot \avgcost(f).$$

Finally, we check that solution $\scSol$ is still not too expensive in the instance $I'$.

\begin{lemma}\label{lem:scSol-Ip}
For every $f \in \scSol$ it holds that
\begin{equation}\label{eq:scSol-Ip-f}
\opencost(f) + \sum_{c \in \cluster'(f)} \dist(c, f) \leq (1+\eps^2) \cdot |\cluster(f)| \cdot \avgcost(f).
\end{equation}
In total, we have
\begin{equation}\label{eq:scSol-Ip}
\opencost(\scSol)+\sum_{f\in \scSol}\,\sum_{c\in \cluster'(f)} \dist(c,f)\leq 2\alpha\eps^{-1}\cdot \OPT.
\end{equation}
\end{lemma}
\begin{proof}
Recall that
$$|\cluster(f)| \cdot \avgcost(f) = \opencost(f) + \sum_{c \in \cluster(f)} \dist(c, f).$$
Thus, to show~\eqref{eq:scSol-Ip-f}, it suffices to prove that
$$\sum_{c \in \Far(f) \cup \Close(f)} \left(\dist(x(f),f)- \dist(c,f))\right) \leq \eps^2 |\cluster(f)| \cdot \avgcost(f).$$
For each $c\in \Far(f)$, we have $\dist(x(f),f)=\eps^2\cdot \avgcost(f)$ and $\dist(c,f)\geq \eps^{-2}\cdot \avgcost(f)$, hence $\dist(x(f),f)-\dist(c,f)\leq 0$.
On the other hand $\dist(x(f),f)=\eps^2\cdot \avgcost(f)$, hence for each $c\in \Close(f)$ we have $\dist(x(f),f)-\dist(c,f)\leq \eps^2\cdot \avgcost(f)$.
This proves~\eqref{eq:scSol-Ip-f}.

By summing~\eqref{eq:scSol-Ip-f} over all $f \in \scSol$ we obtain that
$$\opencost(\scSol)+\sum_{f\in \scSol}\sum_{c\in \cluster'(f)} \dist(c,f)\leq (1+\eps^2) \cdot \cost(\scSol) \leq 2\alpha\eps^{-1} \cdot \OPT,$$
as claimed.
\end{proof}

Note that in Lemma~\ref{lem:scSol-Ip}, the left hand side of~\eqref{eq:scSol-Ip} is lower bounded by $\cost(\scSol,I')$, but is not necessarily equal to it, as the clients of each cluster $\cluster'(f)$ are assigned to $f$,
which may cease to be the closest facility after moving a client.

\paragraph*{Layering on magnitudes of the average cost.} We now work with the instance $I'$. 
The goal is to use the obtained properties of clusters to break the instance into several independent ones at the cost of additionally paying $\eps \OPT$, 
so that each of the instances concerns only clients from clusters with average cost of the same magnitude.
This is because such instances can be solved efficiently using the following crucial lemma, whose proof will be given later.

\begin{lemma}\label{lem:same-rad}
Suppose we are given an instance $J=(G,\clients,\fac,\opencost)$ of {\sc{Non-uniform Facility Location}} where $G$ is planar.
Moreover, we are provided a real $r>1$ and a set of facilities $\hintSol\subseteq \fac$ such that the clients of $\clients$ can be partitioned into nonempty clusters $(\cluster(f))_{f\in \hintSol}$ so that
the following properties hold for each $f\in \hintSol$:
\begin{itemize}
\item $1\leq \dist(c,f)\leq r$ for each $c\in \cluster(f)$; and
\item $\opencost(f)+\sum_{c\in \cluster(f)} \dist(c,f)\leq |\cluster(f)|\cdot r$.
\end{itemize}
Then, given $\eps>0$, one can in time $n^{\Oh(\eps^{-2}r)}$ compute a solution to $J$ with cost at most $(1+\eps)\OPT(J)+\eps\cdot M$,
where $M=\opencost(\hintSol)+\sum_{f\in \hintSol}\sum_{c\in \cluster(f)}\dist(c,f)$.
\end{lemma}

Breaking into separate instances that can be treated using Lemma~\ref{lem:same-rad} will be done using layering on the levels of magnitude of average costs of facilities from $\scSol$.
While the layering itself will be quite standard, the proof of the separation property between the instances will be quite non-trivial and will require the fine understanding of properties of $\scSol$ that 
we have developed.

Let us partition the facilities of $\scSol$ into layers $(\lr_i)_{i\in \Z}$, where $\lr_i$ comprises facilities $f\in \scSol$ satisfying
$$\eps^{4i}\geq \avgcost(f)> \eps^{4i+4}.$$
For $i\in \Z$, let 
$$\ell_i = \sum_{f\in \lr_i} \left(\opencost(f)+\sum_{c\in \cluster'(f)} \dist(c,f)\right).$$
By Lemma~\ref{lem:scSol-Ip}, we have
\begin{equation}\label{eq:sum-layers}
\sum_{i\in \Z} \ell_i\leq 2\alpha\eps^{-1}\cdot \OPT.
\end{equation}
Let $q=\lceil \eps^{-2}\rceil$. Pick $a\in \{0,1,\ldots,q-1\}$ such that $\sum_{i\colon i\equiv a\bmod q} \ell_i$ is minimum. Then by~\eqref{lem:scSol-Ip} and the fact that $q\geq \eps^{-2}$ we infer that
\begin{equation}\label{eq:baker-cost}
\sum_{i\colon i\equiv a\bmod q} \ell_i\leq \eps^2\cdot \cost(\scSol;I')\leq 2\alpha\eps\cdot \OPT.
\end{equation}
Now, define
$$S=\bigcup_{i\colon i\equiv a\bmod q} \lr_i\qquad\textrm{and}\qquad \ring_j = \bigcup_{jq+a<i<(j+1)q+a} \lr_i\quad\textrm{for }j\in \Z.$$
Set $\ring_j$ will be called the {\em{$j$-ring}}. It follows that $S$ and $(\ring_j)_{j\in \Z}$ form a partition of $\scSol$. 

Intuitively, the idea is to construct a near optimum solution by buying all the facilities of $S$ and using them to serve all clients served by them in $\scSol$ (the cost of this is bounded by $2\alpha\eps\cdot \OPT$ 
by~\eqref{eq:baker-cost}), and constructing an instance for each nonempty ring $\ring_j$ that is subsequently approximated using Lemma~\ref{lem:same-rad}. However, we need to prepare those instances carefully so that
they can be solved separately.

To this end, we heavily rely on Lemma~\ref{lem:close-opt-abs} that more or less says that one needs to open a facility within $2\cdot \avgcost(f)$ of $f$ for every $f \in \scSol$.
This, together with the exponential scale of average costs, implies that while
focusing on the ring $\ring_j$ we do not need to understand how the solution to rings $\ring_{j'}$ for $j' > j$ looks like (namely, what are the
precise locations of the facilities); instead, we just put one zero-cost facility at every
$f \in \ring_{j'}$ that mimicks the closest opened facility, this will
be satisfying up to losing a factor $(1+\eps)$.

Let us now proceed with formal details. Denote
$$C_S=\bigcup_{f \in S} \cluster'(f).$$
For every $j \in \mathbb{Z}$ we create the following instance $J_j = (G, \clients_j, \fac_j, \opencost_j)$:
\begin{itemize}
\item The graph $G$ is the graph from the original instance;
\item $\clients_j = \bigcup_{f \in \ring_j} \cluster'(f)$, that is, all clients in clusters of facilities from the ring $\ring_j$;
\item $\fac_j = \fac$ are all facilities from the input;
\item $\opencost_j(f) = 0$ for every $f \in \ring_{j'}$ with $j' > j$ and every $f \in S$, and $\opencost_j(f) = \opencost(f)$ otherwise.
\end{itemize}
Note that the sets $(\clients_j)_{j \in \mathbb{Z}}$ are pairwise disjoint and together with $C_S$ form a partition $\clients$.
For every $j \in \mathbb{Z}$ let $\freefac_j = S \cup \bigcup_{j' > j} \ring_{j'}$ be the set of facilities $f$ with $\opencost_j(f)$ redefined to $0$ in the definition of $J_j$.

Observe also that if $\ring_j = \emptyset$, then $\clients_j = \emptyset$: the instance is trivial and it admits the empty set as the optimum solution.
The algorithm does not really need to construct these instances (and thus in fact constructs at most $n$ instances $J_j$), but we prefer to define them for the sake of clarity of notation.
We henceforth call the instances $J_j$ \emph{trivial} if $\ring_j = \emptyset$ and \emph{nontrivial} otherwise.

We now verify that it suffices to solve each instance $J_j$ separately. This is done through two lemmas. In the first one, we show how to combine
solutions to the instances $J_j$ into a solution to the instance $I'$.
\begin{lemma}\label{lem:JtoI}
Assume we are given sets $D_j \subseteq \fac_j$ for every nontrival instance $J_j$.
Then one can construct in polynomial time a set $D \subseteq \fac$ such that
\begin{equation}\label{eq:JtoI}
\cost(D ; I') \leq \sum_j \cost(D_j; J_j) + 10\alpha\eps \cdot \OPT.
\end{equation}
\end{lemma}
\begin{proof}
For every nontrivial instance $J_j$ and for every $f \in \fac_j \setminus D_j$ we check whether opening $f$ would not decrease the cost of $D_j$ in $J_j$;
if this is the case, we add $f$ to $D_j$. We also add $\freefac_j$ to $D_j$ as it does not increase the cost of $D_j$.
Henceforth we assume that for every nontrivial instance $J_j$ and every $f \in \fac_j \setminus D_j$ it holds that
\begin{equation}\label{eq:JtoI:1opt}
\cost(D_j \cup \{f\}; J_j) \geq \cost(D_j ; J_j).
\end{equation}

We define $D_j = \freefac_j$ for every trivial instance $J_j$.
Note that property~\eqref{eq:JtoI:1opt} also holds for the trivial instances. 
Let $D_j' = D_j \setminus \freefac_j$ for every $j \in \mathbb{Z}$; note that $D_j' = \emptyset$ for trivial $J_j$.
Let
$$D := S \cup \bigcup_{j \in \mathbb{Z}} D_j'.$$
We claim that $D$ satisfies the requirements of the lemma; it is clearly computable in polynomial time as $D_j' = \emptyset$ for trivial $J_j$.
Note that $D_j \setminus D_j' = \freefac_j$ for every $j \in \mathbb{Z}$.

For a facility $f \in D_j$, let $\cluster(f, D_j ; J_j) \subseteq \clients_j$ be the set of clients served by $f$ in the solution $D_j$ to~$J_j$; that is,
$\cluster(f, D_j ; J_j)$ is the set of these $c \in \clients_j$ for which $f$ is the closest facility from $D_j$.
Consider redirecting, in the solution $D$ to the instance $I'$, all clients from $\cluster'(f)$ to $f$, for every $f \in S \subseteq D$. Then we have:
\begin{align*}
\cost(D ; I') &\leq \left(\opencost(S) + \sum_{f \in S}\, \sum_{c \in \cluster'(f)} \dist(c, f) \right) \\
              & + \left(\opencost\left(\bigcup_{j \in \mathbb{Z}} D_j'\right) + \sum_{j \in \mathbb{Z}}\, \sum_{f \in D_j'}\, \sum_{c \in \cluster(f, D_j ; J_j)} \dist(c, D)\right) \\
              & + \left(\sum_{j \in \mathbb{Z}}\, \sum_{f \in \freefac_j} \sum_{c \in \cluster(f, D_j; J_j)} \dist(c, D)\right).
\end{align*}
We bound the three summands in the inequality above separatedly. By~\eqref{eq:baker-cost}, the first summand is bounded by $2\alpha\eps\OPT$.
Since $D_j' \subseteq D \cap D_j$ for every $j \in \mathbb{Z}$, we have for the second summand:
\begin{align*}
&\opencost\left(\bigcup_{j \in \mathbb{Z}} D_j'\right) + \sum_{j \in \mathbb{Z}}\, \sum_{f \in D_j'} \sum_{c \in \cluster(f, D_j ; J_j)} \dist(c, D) \\
&\quad \leq \sum_{j \in \mathbb{Z}} \left( \opencost(D_j) + \sum_{f \in D_j'}\, \sum_{c \in \cluster(f, D_j ; J_j)} \dist(c, f) \right) \\
&\quad = \sum_{j \in \mathbb{Z}} \left( \opencost(D_j) + \conncost\left(\bigcup_{f \in D_j'} \cluster(f, D_j; J_j), D_j; J_j\right) \right).
\end{align*}
We now estimate the third summand.
Consider a nontrivial instance $J_j$ and a facility $f \in \ring_j$. 
Recall that $\cluster'(f) \subseteq \clients_j$. 
By applying Lemma~\ref{lem:close-opt-abs} to the instance $J_j$, solution $D_j$, facility $f$, and set $K = \cluster'(f)$ we infer
that~\eqref{eq:JtoI:1opt} ensures that there exists $g \in D_j$ with
$$\dist(f, g) \leq 2\cdot \frac{\opencost(f) + \sum_{c \in \cluster'(f)} \dist(c, f)}{|\cluster'(f)|}.$$
Plugging now the bound of Lemma~\ref{lem:scSol-Ip}, we obtain
\begin{equation}\label{eq:close-Dj}
\dist(f, D_j) \leq 2(1+\eps^2) \cdot \avgcost(f) \leq 4\cdot \avgcost(f) \leq 4\eps^{4(jq + a+1)}.
\end{equation}
We now observe the following.

\begin{claim}\label{cl:close-to-D}
For every facility $f\in D_j$, we have
$$\dist(f,D)\leq 4\sum_{j' = j + 1}^\infty \eps^{4(j'q+a+1)}.$$
\end{claim}
\begin{clproof}
Since all but a finite number of $D_j$-s are empty, we can proceed by induction on $j$, assuming the claim holds for all $j'>j$.
Take any $f\in D_j$. If $f\in D$ then $\dist(f,D)=0$ and we are done.
Otherwise, $f\in D_j\setminus D\subseteq \bigcup_{j'>j} D_{j'}$, so $f\in D_{j'}$ for some $j'>j$.
By~\eqref{eq:close-Dj}, there exists $g\in D_{j'}$ such that
$$\dist(f,g)\leq 4\eps^{4(j'q + a+1)}.$$
By induction assumption for $g$, we have
$$\dist(g,D)\leq 4\sum_{j'' = j' + 1}^\infty \eps^{4(j''q+a+1)}.$$
Hence, we have
$$\dist(f,D)\leq \dist(f,g)+\dist(g,D)\leq 4\eps^{4(j'q + a+1)}+4\sum_{j'' = j' + 1}^\infty \eps^{4(j''q+a+1)}\leq 4\sum_{j' = j + 1}^\infty \eps^{4(j'q+a+1)},$$
as required.
\end{clproof}

By Claim~\ref{cl:close-to-D}, for every $f \in \freefac_j$ and $c \in \cluster(f, D_j; J_j)$ with $c \in \cluster'(f_c)$ for some $f_c \in \ring_j$ we have the following:
\begin{align*}
\dist(c, D) &\leq \dist(c, f) + \dist(f,D) \\
  &\leq \dist(c, f) + \sum_{j' = j + 1}^\infty 4\eps^{4(j'q+a+1)} \\
  &\leq \dist(c, f) + \eps^{4((j+1)q+a+1)} \cdot \frac{4}{1 - \eps^{4q}} \\
  &\leq \dist(c, f) + 8\eps^4\cdot \avgcost(f_c).
\end{align*}
By summing the above bound through all $j \in \mathbb{Z}$ and $f \in \freefac_j$ we obtain
$$\sum_{j \in \mathbb{Z}}\, \sum_{f \in \freefac_j}\, \sum_{c \in \cluster(f, D_j; J_j)} \dist(c, D) 
\leq \sum_{j \in \mathbb{Z}} \conncost\left(\bigcup_{f \in \freefac_j} \cluster(f, D_j; J_j), D_j; J_j\right) + 8\eps^4 \cdot \cost(\scSol).$$
Since $\cost(\scSol) \leq \alpha\eps^{-1}\cdot \OPT$, we can combine the obtained bounds as follows:
\begin{eqnarray*}
\cost(D ; I') & \leq & 2\alpha\eps\OPT + \sum_{j \in \mathbb{Z}} \left( \opencost(D_j) + \conncost\left(\bigcup_{f \in D_j'} \cluster(f, D_j; J_j), D_j; J_j\right)\right) +\\
&   & \sum_{j \in \mathbb{Z}} \conncost\left(\bigcup_{f \in \freefac_j} \cluster(f, D_j; J_j), D_j; J_j\right) + 8\alpha\eps^3\OPT\\
& = & \sum_j \cost(D_j; J_j)+2\alpha\eps\OPT+8\alpha\eps^3\OPT\leq \sum_j \cost(D_j; J_j)+10\alpha\eps\OPT.
\end{eqnarray*}
This concludes the proof.
\end{proof}

The second lemma shows that optima in instances $J_j$ almost partition the optimum in $I$.

\begin{lemma}\label{lem:ItoJ}
For $j \in \mathbb{Z}$, let $\OPT_j$ be the cost of the optimum solution of $J_j$. 
Then
$$\sum_{j \in \mathbb{Z}} \OPT_j \leq (1+9\alpha\eps)\cdot \OPT.$$
\end{lemma}
\begin{proof}
Let $\openfac$ be an optimum solution to $I'$. 
For every $f \in \openfac$ let $j(f)$ be the maximum value of $j$ such that there exists $g \in \ring_j$ with $\dist(f,g) \leq 3\eps^{-2}\cdot \avgcost(g)$.
If no such $j$ exists, we set $j(f)$ to be the minimum value of $j$ for which $J_j$ is nontrivial. 
For every $j\in \Z$ we define $$D_j' = \{f \in \openfac~|~j(f) = j\}\qquad \textrm{and}\qquad D_j = D_j' \cup \freefac_j;$$ note that $D_j' = \emptyset$ for trivial $J_j$. 
Our goal is to estimate $\sum_{j \in \mathbb{Z}} \cost(D_j; J_j)$ by $\cost(\openfac, I')$ plus some terms of the order of $\eps \cdot \OPT$.
First, it is immediate from the definition that $\opencost(\openfac) = \sum_{j \in \mathbb{Z}} \opencost_j(D_j)$.
Clearly, for trivial $J_j$ we have $D_j = \freefac_j$ and $\cost(D_j; J_j) = 0$. Let $J_j$ be nontrivial. 
Consider a client $c \in \clients_j$; by the definition of $J_j$, there exists $f_0 \in \ring_j$ with $c \in \cluster'(f_0)$.

Let $f \in \openfac$ be the facility that serves $c$ in the solution $\openfac$, that is, $\dist(c, f) = \dist(c, \openfac)$. 
We consider cases depending on the relation of $j(f)$ and $j$.

\medskip

\noindent\textbf{Case 1: $j(f) > j$.} By the definition of $j(f)$, there exists $g \in \ring_{f(j)} \subseteq \freefac_j$ with $\dist(f, g) \leq 3\eps^{-2}\cdot \avgcost(g) \leq 3\eps^2\cdot \avgcost(f_0)$.
Therefore
$$\dist(c, D_j) \leq \dist(c,g)\leq \dist(c,f)+ 3\eps^2\cdot \avgcost(f_0)=\dist(c, \openfac) + 3\eps^2\cdot \avgcost(f_0).$$

\noindent\textbf{Case 2: $j(f) = j$.} Here $f \in D_j$ and thus $\dist(c, D_j) \leq \dist(c,f)= \dist(c, \openfac)$. 

\noindent\textbf{Case 3: $j(f) < j$.}
Supposing that $f_0\notin \openfac$, Lemma~\ref{lem:close-opt-abs} applied to the (optimal) solution $\openfac$ in $I'$ with facility $f_0$ and $K = \cluster'(f_0)$
yields that there exists $g_0 \in \openfac$ with
$$\dist(f_0, g_0) \leq 2\cdot \frac{\opencost(f_0) + \sum_{c \in \cluster'(f_0)} \dist(c, f_0)}{|\cluster'(f_0)|} \leq 2(1+\eps^2)\cdot \avgcost(f_0) \leq 4\cdot \avgcost(f_0).$$
Here, the penultimate inequality follows from Lemma~\ref{lem:scSol-Ip}.
If $f_0\in \openfac$, then we can take $g_0=f_0$ and the above inequality also holds.

By the definition of $j(f)$ we have that $\dist(f, f_0) > 3\eps^{-2}\cdot \avgcost(f_0)$. On the other hand,
  $\dist(c, f_0) \leq \eps^{-2}\cdot \avgcost(f_0)$ while $\dist(f_0, g_0) \leq 4\cdot \avgcost(f_0) \leq \eps^{-2}\cdot \avgcost(f_0)$. Since $g_0 \in \openfac$, we infer
  that $f$ is not the closest to $c$ facility of $\openfac$, a contradiction.  We infer that this case is impossible.

We conclude that in any case, we have 
$$\dist(c,D_j)\leq \dist(c,\openfac)+ 3\eps^2\cdot \avgcost(f_0).$$
By summing this bound through all the clients and adding opening costs to both sides, we obtain 
$$\sum_{j \in \mathbb{Z}} \cost(D_j ; J_j) \leq \cost(\openfac ; I') + 3\eps^2\cdot \cost(\scSol) \leq \OPT' + 3\alpha\eps\cdot\OPT \leq (1+9\alpha\eps)\OPT,$$
where in the last inequality we use Lemma~\ref{lem:moving}.
This finishes the proof of the lemma.
\end{proof}

We conclude this section with the observation that it remains to prove Lemma~\ref{lem:same-rad}
in order to show a polynomial-time approximation scheme for \textsc{Non-Uniform Facility Location} in planar graphs.
After initial preprocessing of the input instance $I$, Corollary~\ref{cor:move-equiv}
asserts that it suffices to find a $(1+\Oh(\varepsilon))$-approximate solution to $I'$.

To this end, we break $I'$ into instances $(J_j)_{j \in \mathbb{Z}}$.
For every nontrivial $J_j$, we scale all the edge lengths and opening costs 
of $J_j$ by a factor of $\eps^{-(4(jq+q+a)+2)}$
and define $\hintSol = \ring_j$ and $\cluster(f) := \cluster'(f)$ for every $f \in \hintSol$.
Note that $(\cluster(f))_{f \in \hintSol}$ partitions $\clients_j$.
Let $$r = 2\eps^{-4q} \leq 2\eps^{-4\eps^{-2}}.$$ Then, since
for every $f \in \ring_j$ we have 
\begin{equation}\label{eq:ringavg}
\eps^{4(jq+a+1)} \geq \avgcost(f) > \eps^{4(jq + q + a)}
\end{equation}
and for every $c \in \cluster'(f)$ it holds that
$$\eps^2\cdot \avgcost(f) \leq \dist(c, f) \leq \eps^{-2}\cdot \avgcost(f),$$
we infer that after scaling the distances, $1 \leq \dist(c, f) \leq r/2$
for every $f \in \ring_j$ and $c \in \cluster'(f)$. 
Furthermore,~\eqref{eq:ringavg} together with Lemma~\ref{lem:scSol-Ip} imply
the second condition of Lemma~\ref{lem:same-rad}.

Consequently, the algorithm Lemma~\ref{lem:same-rad} applied to $J_j$ prepared as above
with accuracy parameter $\eps^2$ (instead of $\eps$)
runs in time $n^{2^{\Oh(\eps^{-2} \log (1/\eps))}}$ and returns a solution $D_j$ 
of cost (after scaling back again all the edge weights and opening costs) satisfying
  $$\cost(D_j; J_j) \leq (1+\eps^2)\OPT_j + \eps^2 \cdot M_j,$$
  where
  $$M_j = \opencost(\ring_j) + \sum_{f \in \ring_j} \sum_{c \in \cluster'(f)} \dist(c, f).$$
Observe that
$$\sum_{j \in \mathbb{Z}} M_j \leq \cost(\scSol) \leq 2\alpha\eps^{-1}\OPT.$$
Thus Lemma~\ref{lem:JtoI} allows us to combine the solutions $D_j$ into a solution $R$
to $I'$ of cost satisfying:
$$\cost(R;I')\leq (1+\eps^2) \sum_{j \in \mathbb{Z}} \OPT_j + 12\alpha\eps \cdot \OPT.$$
By Lemma~\ref{lem:ItoJ}, this value is at most
$$(1+\eps^2)\OPT'+18\alpha\eps \cdot \OPT.$$
Finally, we may apply Corollary~\ref{cor:move-equiv} to conclude that
$$\cost(R;I)\leq (1+2\eps^2+8\alpha\eps)\OPT+18\alpha\eps\cdot \OPT\leq (1+28\alpha\eps)\cdot \OPT,$$
as required. Consequently, it remains to prove Lemma~\ref{lem:same-rad}.

%% file: nufl-dp.tex
\subsection{Overview}
\input{nufl-dp-overview}

\input{nufl-ptas-dp}

%% file: nufl-dp-overview.tex
\newcommand{\req}{\mathsf{req}}
\newcommand{\pred}{\mathsf{pred}}

Before we proceed to the formal proof of Lemma~\ref{lem:same-rad}, we give a short overview. The approach is based on a rather standard layering argument plus portal-based Divide\&Conquer. While the formal reasoning is quite lengthy due to a number of technical details that require attention, we hope that presenting an intuitive description of consecutive steps will help the reader with guiding through the proof.

Suppose $D$ is an optimum solution to instance $J$.
The first realization is that $D$ enjoys a similar proximity property as expressed in Lemma~\ref{lem:close-opt-abs}.
Namely, every client $c\in C$ is at distance at most $3r$ from some facility of $D$.
The argument is essentially the same: supposing all clients from some cluster $\cluster(f)$ for $f\in \hintSol$ had connection costs larger than $r$ in the solution $D$, 
one could improve $D$ by opening facility $f$ and rediricting all clients from $\cluster(f)$ to $f$. Otherwise, some client from $\cluster(f)$ is within distance at most $r$ from $D$, which implies that
all of them are at distance at most $3r$.

This proximity property allows us to apply standard layering.
We fix a vertex $s$ and classify facilities from $\hintSol$ of the graph into layers $(\hintSol_i)_{i\in \N}$ of width $8r$ according to distances from $s$: layer $\hintSol_i$ comprises facilities $f\in \hintSol$ satisfying $i\cdot 8r\leq \dist(s,f)<(i+1)\cdot 8r$. With every facility $f\in \hintSol$ we can associate its contribution to $M$, equal to $\opencost(f)+\sum_{c\in \cluster(f)}\dist(c,f)$.
Now, denoting $q=\lceil \eps^{-1}\rceil$, there exists $a\in \{0,1,\ldots,q-1\}$ such that the total contribution of facilities from layers $\hintSol_i$ with $i\equiv a\bmod q$ is at most $\eps M$.
Hence, by paying cost $\eps M$ we may open these facilities and direct all clients from their clusters to them. Now it is easy to see that we have a separation property: instance $J$ can be decomposed into instances $(J_j)_{j\in \N}$, where $J_j$ concerns connecting all clients from clusters of facilities of $\bigcup_{jq+a<i<(j+1)q+a} \hintSol_i$ to facilities within distance between $(jq+a)\cdot 8r -4r$ and $((j+1)q+a)\cdot 8r -4r$ from $s$, which can be (approximately) solved separately. This is because in the optimum solution, no client-facility path used for connection crosses any of the entirely bought layers due to having length at most $3r$.

Let us focus on one instance $J_j$. We may contract all vertices at distance less than $(jq+a)\cdot 8r-8r$ onto $s$ and remove all vertices at distance more than $((j+1)q+a)\cdot 8r$, as these vertices anyway will not participate in any shortest path used by an optimum solution. Thus, we essentially achieve a small radius property in $J_j$: one may assume that all vertices are at distance at most $8qr=\Oh(\eps^{-1}r)$ from $s$.

The idea is to compute a near-optimum solution to $J_j$ using Divide\&Conquer on balanced separators, presented as dynamic programming. Using standard separation properties of planar graphs one can show that the graph (or rather its plane embedding) admits a hierarchical decomposition into regions so that the decomposition has depth at most $\log n$ and every region is boundaried by a union of at most $6$ shortest paths, all with one endpoint in $s$. Thus, each of these paths has length $\Oh(\eps^{-1}r)$. We apply dynamic programming over this decomposition, where we put portals on the boudaries of regions to limit the number of states. That is, along each path we put portals spaced at $\delta$, for some parameter $\delta>0$, and we allow paths connecting clients with facilities to cross region boundaries only through portals. Since the decomposition has depth $\log n$, each connection path in the optimum solution can be ``snapped to portals'' to conform with this requirement by using at most $2\log n$ snappings, incurring a total additional cost of $2\delta\cdot \log n$.
Therefore, we put $\delta=\eps/\log n$ so that this error is bounded by $\Oh(\eps)$, which summed through all clients yields an $\Oh(\eps M)$ error term in total. Thus, the total number of portals on the boundary of each region is $\Oh(\delta^{-1}\eps^{-1}r)=\Oh(\eps^{-2}r\log n)$.

In the dynamic programming state associated with a region $R$, we are concerned about opening facilities within $R$ to serve all clients in $R$. However, on the boundary of $R$ we have $\Oh(\eps^{-2}r\log n)$ portals that carry information about the assumed interaction between the parts of the overall solution within $R$ and outside of $R$. For every portal $\pi$, this information consists of two pieces:
\begin{itemize}
\item {\em{request}} $\req(\pi)$ that gives a hard request on the sought solution within $R$: there has to be a facility opened at distance at most $\req(\pi)$ from $\pi$;
\item {\em{prediction}} $\pred(\pi)$ that gives a possibility of connecting clients to portals: every client $c$ can be connected to $\pi$ at connection cost $\dist(c,\pi)+\pred(\pi)$.
\end{itemize}
Intuitively, predictions represent ``virtual'' opened facilities residing outside of $R$, which can be accessed at an additional cost given by $\pred(\pi)$, while by satisfying requests we make sure that predictions in other regions can be fulfilled. Since all client-facility paths in the optimum solution are of length at most $3r$, we may assume that all requests and predictions in all considered states are bounded by $3r$.
At the cost of an additional error term $\Oh(\eps M)$ we can also assume that requests and predictions are rounded to integer multiples of $\delta$.
Thus, for every portal $\pi$ we can limit ourselves to $\Oh(\delta^{-1}r)=\Oh(\eps^{-2}r\log n)$ possibilities for $\req(\pi)$ and same for $\pred(\pi)$.

Let us estimate the number of states constructed so far. For each of $\Oh(\eps^{-2}\log n)$ portals on the boundary of $R$ we have $\Oh(\eps^{-2}r\log n)$ possibilities for $\req(\pi)$ and for $\pred(\pi)$, yielding a total number of states being $(\eps^{-2}r\log n)^{\Oh(\eps^{-2}r\log n)}=n^{\textrm{poly}(1/\eps)\cdot r\cdot \log \log n}$, which is quasi-polynomial. As transitions in this dynamic programming can be implemented efficiently, this already yields a QPTAS, and we are left with reducing the number of states to polynomial.

The final trick is to take a closer look at what we store in the states. Since $\req(\cdot)$ stores the requested distance to the closest facility opened within $R$, it is safe to assume that $\req(\cdot)$ (before rounding to integer multiples of $\delta$) will be $1$-Lipschitz in the following sense: for any two portals $\pi,\rho$, we have
$$|\req(\pi)-\req(\rho)|\leq \dist(\pi,\rho).$$
An analogous reasoning can be applied to predictions, so we can assume that $\pred(\cdot)$ is $1$-Lipschitz as well.
Now consider any of the $6$ shortest paths comprising the boundary of $R$, say $P$. On this path we put portals spaced at $\delta$, say $\pi_1,\ldots,\pi_\ell$ for $\ell\leq \Oh(\eps^{-2}r\log n)$ in the order on $P$.
As argued, after rounding we have $\Oh(\eps^{-2}\log n)$ possibilities for $\req(\pi_1)$, but observe that once (rounded) $\req(\pi_{i-1})$ is chosen, there are only at most $5$ possibilites for $\req(\pi_i)$: it must be an integer multiple of $\delta$ that differs from $\req(\pi_{i-1})$ by at most $2\delta$, due to $\dist(\pi_{i-1},\pi_i)=\delta$. 
Hence, the total number of choices for the values of requests along $P$ is bounded by $\Oh(\eps^{-2}\log n)\cdot 5^{\Oh(\eps^{-2}r\log n)}=n^{\Oh(\eps^{-2}r)}$.
Same argument applies to predictions, and as the boundary of $R$ consists of at most $6$ such paths, the total number of states we need to consider is $n^{\Oh(\eps^{-2}r)}$.

%% file: nufl-ptas-dp.tex
\subsection{Proof of Lemma~\ref{lem:same-rad}} 

We now proceed with the formal proof of Lemma~\ref{lem:same-rad}.
For the remainder of this section, let us fix the setting and the notation from the statement of Lemma~\ref{lem:same-rad}.

Fix an optimum solution $\openfac\subseteq \fac$ in the instance $J$.
We first prove that in fact, every client is not too far from its closest facility in $\openfac$.

\begin{lemma}\label{lem:close-openfac}
For each $c\in \clients$ there exists $g\in \openfac$ such that $\dist(c,g)\leq 3r$.
\end{lemma}
\begin{proof}
Let $f\in \hintSol$ be such that $c\in \cluster(f)$; then $\dist(c,f)\leq r$.
We shall prove that there exists some client $d\in \cluster(f)$ and facility $g\in \openfac$ such that $\dist(d,g)\leq r$.
Indeed, if this is true, then we have $\dist(c,g)\leq \dist(c,f)+\dist(f,d)+\dist(d,g)\leq r+r+r=3r$, as required.

Suppose otherwise: for each $d\in \cluster(f)$, the distance from $d$ to the closest facility from $\openfac$ is larger than $r$.
As $\cluster(f)$ is nonempty, the total connection cost incurred by clients from $\cluster(f)$ in solution $\openfac$ can be lower bounded as follows:
$$\sum_{c\in \cluster(f)} \dist(c,\openfac)>|\cluster(f)|\cdot r\geq \opencost(f)+\sum_{c\in \cluster(f)} \dist(c,f).$$
This means that the solution $\openfac\cup \{f\}$ has a strictly smaller cost than $\openfac$, which contradicts the optimality of $\openfac$.
\end{proof}

Let $G'$ be the subgraph of $G$ induced by all vertices whose distance from $\hintSol$ is at most $4r$. 
Observe that all clients of $\clients$ are placed at vertices of $G'$.
Lemma~\ref{lem:close-openfac} now immediately implies the following.

\begin{lemma}\label{lem:trim-openfac}
It holds that $\openfac\subseteq V(G')$ and for every $c\in \clients$ we have $\dist_{G'}(c,\openfac)=\dist_G(c,\openfac)$.
\end{lemma}
\begin{proof}
For the first assertion, by the optimality of $\openfac$, for every $g\in \openfac$ there exists some client $c\in \clients$ such that $g$ is the facility of $\openfac$ closest to $c$.
By Lemma~\ref{lem:close-openfac} we have $\dist_G(c,g)\leq 3r$. If now $f\in \hintSol$ is such that $c\in \cluster(f)$, then $\dist_G(c,f)\leq r$. Hence $\dist_G(f,g)\leq r+3r=4r$, so $g\in V(G')$.

For the second assertion, observe that by Lemma~\ref{lem:close-openfac}, for every client $c\in \clients$, the shortest path from $c$ to a facility of $\openfac$ traverses only vertices that are at distance at most $4r$ from
the facility $f\in \hintSol$ satisfying $c\in \cluster'(f)$. It follows that the distance from $c$ to $\openfac$ is the same in $G$ as in $G'$
\end{proof}

Let $\fac'$ consist of all the facilities that are placed at vertices of $G'$, and let $J'=(G',\clients,\fac',\opencost)$.
We observe that Lemma~\ref{lem:close-openfac} implies that we can work with instance $J'$ instead of $J$.

\begin{corollary}\label{cor:trim}
For every $R\subseteq \fac'$, we have $\cost(R;J')\geq \cost(R;J)$.
Moreover, we have $\cost(\openfac;J')=\cost(\openfac;J)$ and consequently $\OPT(J')=\OPT(J)$.
\end{corollary}
\begin{proof}
The first assertion is straightforward, because $G'$ is an induced subgraph of $G$, hence distances between vertices of $G'$ are not smaller in $G'$ than in $G$.
For the second assertion, observe that by Lemma~\ref{lem:close-openfac} we have $\openfac\subseteq \fac'$ and $\dist_{G'}(c,\openfac)=\dist_G(c,\openfac)$ for every client $c\in \clients$,
hence the connection cost of $\openfac$ in $J$ and in $J'$ are the same.
As the opening costs are also obviously the same, we conclude that indeed $\cost(\openfac;J')=\cost(\openfac,J)$.
This, together with the first assertion, immediately entails $\OPT(J')=\OPT(J)$.
\end{proof}

From now on we will assume that the graph $G'$ is connected.
This can be achieved either by connecting the connected components using edges of very large (but finite) weight, 
or applying the forthcoming reasoning to every connected component of $G'$ separately and taking the union of obtained
solutions.

Fix any vertex $s$ and partition the vertices of $G'$ into layers $(\layer_i)_{i\in \N}$ as follows:
for $i\in \N$ we set:
$$\layer_i = \{v\in V(G')\colon i\cdot 8r \leq \dist(u,s) < (i+1) 8r\}.$$
Let $\hintSol_i = \hintSol\cap \layer_i$. Denote $q=\lceil \eps^{-1}\rceil$. Since $(\hintSol_i)_{i\in \N}$ is a partition of $\hintSol$, 
it follows that there exists $a\in \{0,1,\ldots,q-1\}$ such that denoting $S=\bigcup_{i\colon i\equiv a\bmod q} \hintSol_i$, we have
\begin{equation}\label{eq:isolation}
\sum_{f\in S} \left(\opencost(f)+\sum_{c\in \cluster(f)} \dist(c,f)\right) \leq \eps\cdot \sum_{f\in \hintSol} \left(\opencost(f)+\sum_{c\in \cluster(f)} \dist(c,f)\right)=\eps\cdot M.
\end{equation}
Moreover, obviously such $a$ can be found in polynomial time. For $j\in \N$, define the {\em{$j$-th ring}} as
$$\ring_j = \bigcup_{jq+a<i<(j+1)q+a} \layer_i.$$
For future reference, we note that rings are separated from each other.

\begin{lemma}\label{lem:rings-separated}
For any different $j,j'\in \N$ and $u\in \ring_j$ and $u'\in \ring_{j'}$, we have $\dist_{G'}(u,u')>8r$.
\end{lemma}
\begin{proof}
By the definition of $\ring_j$ and $\ring_{j'}$ and since $j\neq j'$, we have $|\dist_{G'}(u,s)-\dist_{G'}(u',s)|>8r$. Then the statement follows by triangle inequality.
\end{proof}

The idea now is to buy the facilities of $S$ and connect the clients from $\clients_S=\bigcup_{f\in S} \cluster(f)$ 
to the centers of their clusters --- which incurs cost at most $\eps\cdot M$ by~\eqref{eq:isolation} ---
and to construct a separate instance for each ring $\ring_j$ so that these instances can be solved independently. 
We now carefully define those instances.

\newcommand{\wH}{\widetilde{H}}

Fix $j\in \N$ and construct graph $H_j$ obtained from $G'$ in the following manner:
\begin{enumerate}
\item Remove all vertices $w$ of $G'$ satisfying $w\in \bigcup_{\iota > jq+a} L_{\iota}$.
\item Contract all vertices $w$ of $G'$ satisfying $w\in \bigcup_{\iota < (j-1)q+a} L_{\iota}$ onto $s$; we shall use the name $s$ also for the vertex obtained as the result of this contraction.
\item For every vertex $w$ that, after the contraction explained above, becomes a neighbor of $s$, we assign the edge $sw$ weight $\dist_{G'}(s,w)$.
\end{enumerate}
Note that in the second, the set of vertices $w$ contracted onto $s$ induces a connected subgraph of $G'$, and thus the contraction is well-defined and preserves the planarity.
We shall identify vertices of $H_j$ with their origins in $G'$ in the obvious way.

In essence, graph $H_j$ retains all the relevant information about distances between vertices of $\ring_j$. This is formalized in the following lemma.

\begin{lemma}\label{lem:H-retains-distances}
The following assertions hold for each $j\in \N$:
\begin{enumerate}[label=(P\arabic*),ref=(P\arabic*)]
\item\label{p:shrink} For every pair of vertices $u,v\in V(H_j)$, we have $\dist_{H_j}(u,v)\geq \dist_{G'}(u,v)$.
\item\label{p:froms} For every vertex $u\in V(H_j)$, we have $\dist_{H_j}(u,s)=\dist_{G'}(u,s)$.
\item\label{p:ring} For every pair of vertices $u,v\in V(H_j)$ satisfying $u\in \ring_j$ and $\dist_{G'}(u,v)\leq 3r$, we have $\dist_{H_j}(u,v)=\dist_{G'}(u,v)$.
\end{enumerate}
\end{lemma}
\begin{proof}
For assertion~\ref{p:shrink}, it suffices to observe that 
every path in $H_j$ with endpoints $u$ and $v$ can be lifted to a path in $G'$ of the same length by substituting any edge incident to $s$, say $sw$,
by the shortest path between $s$ and $w$ in $G'$. 
For assertion~\ref{p:froms}, we already know that $\dist_{H_j}(u,s)\geq \dist_{G'}(u,s)$, and to see that $\dist_{H_j}(u,s)\leq \dist_{G'}(u,s)$ we may observe that on the shortest path in $G'$ from $s$ to $u$, vertices
contracted onto $s$ form a prefix; this prefix can be then replaced by a single edge of the same weight.
For assertion~\ref{p:ring}, the assumption that $u\in \ring_j$ implies that in $G'$, the vertex $u$ is at distance more than $3r$ from any vertex that is removed or contracted onto $s$ in the construction of $H_j$.
Hence, the shortest path from $u$ to $v$ in $G'$ survives the construction of $H_j$ intact.
\end{proof}

Fix
$$L=8r(q+1)\leq 16\eps^{-1}r.$$
For future reference, we also note the following observation.

\begin{lemma}\label{lem:piercing-ring}
Let $Q$ be a shortest path in $H$ from $s$ to some vertex $u$.
Then the length of $Q-s$ (i.e. $Q$ with the first vertex removed) is smaller than $L$.
\end{lemma}
\begin{proof}
Let $v$ be the successor of $s$ on the path $Q$. By the construction of $H$ we have that $u,v\in \bigcup_{(j-1)q+a\leq \iota\leq jq+a} L_{\iota}$ which in particular means that
$$8r((j-1)q+a)\leq \dist(s,v),\dist(s,u)<8r(jq+a+1).$$
Since $v$ lies on the shortest path from $s$ to $u$, it follows that the length of the suffix of $Q$ from $v$ to $u$ (which is $Q-s$) is equal to the $\dist(v,u)$, 
which in turn is smaller than $8r(jq+a+1)-8r((j-1)q+a)=8r(q+1)=L$.
\end{proof}

Having defined the graph $H_j$, we define the facility set $\fac_j$ and client set $\clients_j$ as follows:
$$\fac_j = \fac'\cap \bigcup_{(j-1)q+a\leq \iota\leq jq+a} L_{\iota}\qquad\textrm{and}\qquad \clients_j=\bigcup_{f\in \hintSol\cap \ring_j} \cluster(f).$$
Note that $\fac_j\subseteq V(H_j)$ and $\clients_j\subseteq V(H_j)$.
Finally, we put
$$J_j=(H_j,\clients_j,\fac_j,\opencost);$$
that is, the opening costs are inherited from the original instance $J$.
We now prove that by paying a small cost, we may solve instances $J_j$ separately.

\begin{lemma}\label{lem:layering-separation}
We have
$$\OPT(J')\geq \sum_{j\in \N} \OPT(J_j).$$
Moreover, for any sequence of solutions $(R_j)_{j\in \N}$ to instances $(J_j)_{j\in \N}$, respectively, we have
$$\cost\left(S\cup \bigcup_{j\in \N} R_j;J'\right)\leq \eps\cdot M+\sum_{j\in \N} \cost(R_j;J_j).$$
\end{lemma}
\begin{proof}
For each $j\in \N$, let $\openfac_j$ be the set consisting of all facilities $f\in \openfac$ with the following property: 
there exists a client $c\in \clients_j$ for which $f$ is the closest facility from $\openfac$.
By Lemmas~\ref{lem:close-openfac} and~\ref{lem:trim-openfac}, we have $\dist_{G'}(c,\openfac_j)\leq 3r$ for all $c\in \clients_j$, while from the definition of
$\openfac_j$ it further follows that $\dist_{G'}(f,\clients_j)\leq 3r$ for all $f\in \openfac_j$.
Also, every client $c\in \clients_j$ is at distance at most $r$ from the center of its cluster, which is a facility of $\hintSol$ that resides in $\ring_j$.
Hence, every facility $f\in \openfac_j$ is at distance at most $4r$ from $\ring_j$.
By Lemma~\ref{lem:rings-separated} and triangle inequality we now infer that sets $(\openfac_j)_{j\in \N}$ are pairwise disjoint.
Moreover, we have $\openfac_j\subseteq \fac_j$ and thus $\openfac_j$ can be treated as a solution to the instance $J_j$
Therefore, by Lemma~\ref{lem:H-retains-distances}, assertions~\ref{p:shrink} and~\ref{p:ring}, we have
\begin{eqnarray*}
\OPT(J') = \cost(\openfac;J') & = & \opencost(\openfac)+\sum_{c\in \clients} \dist_{G'}(c,\openfac) \\
                              & = & \sum_{j\in \N} \left(\opencost(\openfac_j)+\sum_{c\in \clients_j} \dist_{G'}(c,\openfac_j)\right)\\
                              & = & \sum_{j\in \N} \left(\opencost(\openfac_j)+\sum_{c\in \clients_j} \dist_{H_j}(c,\openfac_j)\right) = \sum_{j\in \N} \cost(\openfac_j;J_j) \geq \sum_{j\in \N} \OPT(J_j),
\end{eqnarray*}                     
completing the proof of the first assertion.

For the second assertion, since $\clients_S$ and $(\clients_j)_{j\in \N}$ form a partition of $\clients$, we have
\begin{eqnarray*}
\cost\left(S\cup \bigcup_{j\in \N} R_j;J'\right) & \leq  & \opencost(S) + \sum_{c\in \clients_S} \dist(c,S)+ \sum_{j\in \N}\left(\opencost(R_j) + \sum_{c\in \clients_j} \dist_{G'}(c,R_j)\right) \\
                                                 & \leq  & \opencost(S) + \sum_{c\in \clients_S} \dist(c,S)+ \sum_{j\in \N}\left(\opencost(R_j) + \sum_{c\in \clients_j} \dist_{H_j}(c,R_j)\right) \\
                                                 & \leq  & \eps\cdot M + \sum_{j\in \N} \cost(R_j;J_j).
\end{eqnarray*}
where in the second inequality we use Lemma~\ref{lem:H-retains-distances}, assertion~\ref{p:shrink}, while in the last inequality we use~\eqref{eq:isolation}.
\end{proof}

Hence, from now on we focus on finding a near-optimum solutions to instances $J_j$, for each $j\in \N$ for which $\clients_j\neq \emptyset$, as such solutions can be combined into a near-optimum solution to
$J'$ using Lemma~\ref{lem:layering-separation}, which is then a near-optimum solution to $J$ by Corollary~\ref{cor:trim}. 
This will be done by dynamic programming. Fix $j\in J$ for which $\clients_j$ is non-empty.
For brevity, in the following we write $H$ for $H_j$.
Before we proceed, let us observe that $J_j$ enjoys the same proximity property as $J$, expressed in Lemma~\ref{lem:close-openfac}.

\begin{lemma}\label{lem:close-openfac-Jj}
Suppose $\openfac_j$ is an optimum solution in the instance $J_j$.
Then for each $c\in \clients_j$ there exists $g\in \openfac_j$ such that $\dist_H(c,g)\leq 3r$.
\end{lemma}
\begin{proof}
Apply the same reasoning as in the proof of Lemma~\ref{lem:close-openfac}, noting that all relevant vertices and paths are completely contained $H$ due to being at distance at most $3r$ from $W_j$.
\end{proof}

\paragraph*{Getting a suitable decomposition.}
Our dynamic programming will work over a suitable decomposition of the graph $H$.
To define this decomposition, we will need some structural understanding of $H$ and its embedding.

Recall that we assume that $H$ is embedded in a sphere $\Sigma$.
We shall assume that $H$ is triangulated, as we can always triangulate it using edges of weight $+\infty$.
Let $L$ be the set of faces\footnote{We use $L$ here instead of usual $F$ in order to avoid using the same letter as for facility sets.} of $H$.
For future reference, we let $\xi\colon V(H)\to L$ be a function that assigns to every vertex $u$ of $H$ an arbitrary face $\xi(u)$ incident to $u$.
 
Let $S$ be the spanning tree of shortest paths from $s$. That is, if for each $v\in V(H)$ by $P_v$ we denote the shortest path from $v$ to $s$ in $H$, then $S$ is the union of paths $\{P_v\colon v\in V(H)\}$.
Let $S^\star$ be the spanning subgraph of the dual $H^\star$ of $H$ consisting of edges of $H^\star$ that are dual to the edges {\em{not}} belonging to $S$. 
It is well-known that $S^\star$ is then a spanning tree of $H^\star$.

Let 
$$A=\{(f,g),(g,f)\colon fg\in E(S^\star)\};$$
that is, for each edge $fg$ of $S^\star$ we add to $A$ two (oriented) arcs: $(f,g)$ and $(g,f)$.
For an arc $a\in A$, let $L(a)\subseteq L$ denote the set of those faces of $H$ that are contained in this connected component of $S^\star$ with (unoriented) $a$ removed that contains the head of $a$.
For nonempty $B\subseteq A$, we denote
$$L(B)=\bigcap_{a\in B} L(a),$$
and we put $L(\emptyset)=L$ by convention.
We may now state and prove the decomposition lemma that we shall need; in the following, all logarithms are base $2$.

\newcommand{\bag}{\beta}
\newcommand{\chld}{\mathsf{chld}}
\newcommand{\desc}{\mathsf{desc}}

\begin{lemma}\label{lem:tree-decomp}
In polynomial time one can compute a rooted tree $T$ together with a labelling $\bag$ of nodes of $T$ with subsets of $A$ such that the following holds:
\begin{enumerate}[label=(T\arabic*),ref=(T\arabic*)]
\item\label{c:depth} $T$ has depth at most $\log n$;
\item\label{c:paths} for each node $t$ of $T$, we have $|\bag(t)|\leq 3$;
\item\label{c:root}  if $t_0$ is the root of $T$, then $L(\bag(t_0))=L$;
\item\label{c:leaf}  for each leaf $t$ of $T$, we have $|L(\bag(t))|=1$;
\item\label{c:node}  each non-leaf node $t$ of $T$ has at most $7$ children, and if $\chld(t)$ denotes the set of children of $t$, then 
$$L(\bag(t))=\biguplus_{t'\in \chld(t)} L(\bag(t))\qquad\textrm{and}\qquad\bag(t)\subseteq \bigcup_{t'\in \chld(t)} \bag(t').$$
\end{enumerate}
\end{lemma}
\begin{proof}
A subset $X$ of nodes of $S^\star$ is {\em{connected}} if it induces a connected subtree of $X$.
For a subset of nodes~$X$, by $\prt X$ we denote the set of edges of $S^\star$ with one endpoint in $X$ and second outside of $X$.
Let a {\em{block}} be any nonempty, connected subset of nodes $X$ such that $|\prt X|\leq 3$. Note that since $H$ is triangulated, $S^\star$ is a tree with maximum degree at most $3$, 
so every node of $T$ constitutes a single-node block.

We observe the following.

\begin{claim}\label{cl:block-partition}
Every block $X$ with $|X|\geq 2$ admits a partition into at most $7$ blocks, each of size at most $|X|/2$.
\end{claim}
\begin{clproof}
Let $Z\subseteq X$ be the set of all the nodes of $X$ that have a neighbor (in $S^\star$) outside of $X$.
Then $|Z|\leq 3$ and, consequently, there exists a node $x\in X$ such that every connected component of $S^\star[X]-x$ contains at most one node of $Z$.
Further, it is well known that in $S^\star[X]$ there exists a balanced node: a node $y$ such that every connected component of $S^\star[X]-y$ has at most $|X|/2$ nodes.
Then $S^\star[X]-\{x,y\}$ has at most $5$ connected components, and it is straightforward to see that each of them is a block and contains at most $|X|/2$ nodes.
Hence, as $|X|\geq 2$, for the promised partition of $X$ into blocks we can take the node sets of the connected components of $S^\star[X]-\{x,y\}$, plus blocks $\{x\}$ and $\{y\}$ (or just $\{x\}$, in case $x=y$).
\end{clproof}

We now construct the tree $T$ together with labeling $\bag(\cdot)$ by recursively applying Claim~\ref{cl:block-partition} as follows. We start with the block $L$ and, as long as the currently decomposed block $X$ has size larger than $1$, we apply Claim~\ref{cl:block-partition} to $X$ and recursively decompose all the blocks comprising the obtained partition. Then $T$ is the tree of this recursion and the nodes of $T$ can be naturally labelled with blocks decomposed in corresponding calls; thus, the root of $T$ is labelled by $L$, while the leaves of $T$ are labelled by single-node blocks.
Finally, for every node $t$ of $T$, say associated with a block $X_t$, we set $\beta(t)$ to consist of edges of $\prt X_t$ oriented towards endpoints belonging to $X_t$.
It is straightforward to verify that the obtained pair $(T,\beta)$ satisfies all of the required properties. Also, the above reasoning can be trivially translated into a polynomial-time algorithm computing $(T,\beta)$.
\end{proof}

Thus, Lemma~\ref{lem:tree-decomp} essentially provides a hierarchical decomposition of the face set of $H$ using separators consisting of six-tuples of shortest paths originating in $s$: two per each arc in $\bag(t)$.
The idea is to put portals on those separators and run a bottom-up dynamic programming on the tree $T$ that assembles a near-optimum solution while snapping paths to the portals along the way.
First, however, we need to understand how to put portals on paths in $H$.

\newcommand{\Portals}{\Pi}

\paragraph*{Portalization.} Let $X$ be a set of vertices of $H$ and let $f\colon X\to \R\cup \{+\infty\}$ be a function.
For positive reals $d,\sigma$ and reals $\alpha\leq \beta$, we shall say that $f$ is
\begin{itemize}
\item {\em{$d$-discrete}} if all its values are integer multiples of $d$;
\item {\em{$[\alpha,\beta]$-bounded}} if every its value is either $+\infty$ or belongs to the interval $[\alpha,\beta]$; and
\item {\em{Lipschitz with slack $\sigma$}} if
$$|f(u)-f(v)|\leq \dist(u,v)+\sigma\qquad\textrm{for all }u,v\in X\textrm{ with }f(u),f(v)<+\infty.$$
\end{itemize}
A function that is $d$-discrete, $[\alpha,\beta]$-bounded, and Lipschitz with slack $\sigma$ will be called {\em{$(d,\alpha,\beta,\sigma)$-normal}}.

For portalization of shortest paths we shall use the following lemma.

\begin{lemma}\label{lem:portalization}
Let $P$ be a shortest path in $H$ with one endpoint in $s$ and let $d\in \R_{\geq 0}$.
Then one can find a set $\Portals$ of at most $(L/d)+2$ vertices on $P$ with the following property:
for every vertex $u$ on $P$, there exists $\pi\in \Portals$ such that $\dist(u,\pi)\leq d$.
Moreover, for any reals $\alpha \leq \beta$, the number of functions on $\Portals$ that are $(d,\alpha,\beta,d)$-normal 
is at most $((\beta-\alpha)/d)^2\cdot 2^{\Oh(L/d)}$, and such functions can be enumerated in time $((\beta-\alpha)/d)^2\cdot 2^{\Oh(L/d)}$.
\end{lemma}
\begin{proof}
Let $m=\beta-\alpha$.
Let $P'=P-s$, i.e., $P'$ is $P$ with the first vertex removed.
Then, by Lemma~\ref{lem:piercing-ring}, the length of $P'$ is smaller than $L$.

Let $u$ and $v$ be the endpoints of $P'$; then $P'$ is the shortest path connecting $u$ and $v$.
Partition the vertices of $P'$ into intervals $I_0,I_1,I_2,\ldots,I_{p}$, where $p=\lfloor L/d\rfloor$ such that $I_i$ comprises vertices $w$ of $P'$ satisfying $id\leq \dist(u,w)<(i+1)d$;
since the length of $P'$ is smaller than $L$, each of the vertices of $P'$ is placed in one of these intervals. Observe that vertices within every interval $I_i$ are pairwise at distance smaller than $d$.
Therefore, we may construct a suitable set $\Portals'$ for the path $P'$ by taking one vertex $\pi_i$ from each interval $I_i$ that is non-empty; thus, $\Portals'$ has size at most $p\leq (L/d)+1$.
Finally, we set $\Portals=\Portals'\cup \{s\}$.

We now bound the number $(d,\alpha,\beta,d)$-normal functions $f$ on $\Portals$.
Note that there are at most $m/d+2$ possibilities for the value $f(s)$, as this value is either an integer multiple of $d$ between $\alpha$ and $\beta$, or $+\infty$.
Therefore, it suffices to bound the number of $(d,\alpha,\beta,d)$-normal functions on $\Portals'$ by $(m/d)\cdot 2^{\Oh(L/d)}$.
Recall that $|\Portals'|\leq (L/d)+1$, hence there are at most $2^{(L/d)+1}$ choices on which portals will be assigned value $+\infty$.
Supposing that this choice has been made, we bound the number of choices of (finite) values on remaining portals.
Let $1\leq i_1<i_2<\ldots<i_q\leq p$ be the indices such that portals chosen to be assigned finite values are in intervals $I_{i_1},\ldots,I_{i_q}$.
As above, there are at most $m/d+1$ possibilities for the value $f(\pi_{i_1})$.
However, for $j>1$,  the value $f(\pi_{i_j})$ must satisfy inequality
$$|f(\pi_{i_{j}})-f(\pi_{i_{j-1}})|\leq \dist(\pi_{i_{j}},\pi_{i_{j-1}})+d<(i_j-i_{j-1}+1)d+d=(i_j-i_{j-1})d+2d.$$
As $f(\pi_{i_j})$ has to be an integer multiple of $d$, once $f(\pi_{i_{j-1}})$ has been chosen, there are at most $2(i_j-i_{j-1})+4$ choices for the value of $f(\pi_{i_j})$.
Hence, having chosen $f(\pi_{i_1})$, the number of choices for the remaining values $f(\pi_{i_2}),\ldots,f(\pi_{i_q})$ is bounded by
$$\prod_{j=2}^q (2(i_j-i_{j-1})+4)\leq 6^q\cdot \prod_{j=2}^q (i_j-i_{j-1})\leq 6^q\cdot \prod_{j=2}^q 2^{i_j-i_{j-1}}=6^q \cdot 2^{i_q-i_1}\leq 6^q\cdot 2^p\leq 12^p.$$
Since $p\leq (L/d)+1$, we conclude that the total number of $(d,\alpha,\beta,d)$-normal functions on $\Portals'$ is bounded by $(m/d)\cdot 2^{\Oh(L/d)}$, as required.

The above reasoning can be trivially used to construct the promised enumeration algorithm.
\end{proof}

\newcommand{\lclients}{\clients^{\diamond}}
\newcommand{\lfac}{\fac^{\diamond}}
\newcommand{\Nn}{\mathcal{N}}
\newcommand{\Mm}{\mathcal{M}}
\newcommand{\slack}{\lambda}

\paragraph*{Defining subproblems.} As expected, in dynamic programming we will need to solve more general subproblems, where portals on boundaries of these subproblems are taken into account.
Formally, in an instance of the generalized problem we are working with:
\begin{itemize}
\item The original set of available facilities $\fac_j$, which we denote $\lfac$ for consistency; this set is always the same in all instance of the generalized problem, and is equipped with the original opening
cost function $\opencost(\cdot)$.
\item A subset of relevant clients $\lclients\subseteq \clients_j$; this set varies in instances of the generalized problem.
\item A set of portals $\Portals$, which are vertices of $H$.
\item A {\em{prediction function}} $\pred\colon \Portals\to \R\cup \{+\infty\}$.
\item A {\em{request function}} $\req\colon \Portals\to \R\cup \{+\infty\}$.
\end{itemize}
Whenever considering an instance of the generalized problem, all distances are measured in $H$.
Note that we allow negative requests and predictions.

Consider an instance $K=(\lclients,\Portals,\req,\pred)$ of the generalized problem.
For a solution $R\subseteq \lfac$, the connection cost of a client $c\in \lclients$ is defined as
$$\conncost_K(c,R)=\min(\min_{f\in R} \dist(c,f),\min_{\pi\in \Portals} (\dist(c,\pi)+\pred(\pi))).$$
That is, every client can be connected either to a facility of $f$ at the cost of the distance to this facility, 
or to a portal at the cost of the distance to this portal plus its prediction. Note that portals are always all open, 
so the factor $\min_{\pi\in \Portals} (\dist(c,\pi)+\pred(\pi))$ is independent of the solution $R$. 
We will say that $c$ is {\em{served}} by the facility $f$ or portal $\pi$ for which the minimum above is attained.

A solution $R\subseteq \lfac$ is {\em{feasible}} if for every 
portal $\rho\in \Portals$ with $\req(\rho)\neq +\infty$, its request is satisfied in the following sense: 
$$\min_{f\in R} \dist(\rho,f)\leq \req(x).$$
Note that the request of a portal has to be satisfied by a facility included in the solution; it cannot be satisfied by another portal.
Again $\rho$ is {\em{served}} by the facility $f$ for which the minimum above is attained.

To analyze the approximation error, we will need to gradually relax the feasibility constraint. For this, 
for a nonnegative real $\slack$ we shall say that a solution $R\subseteq \lfac$ is {\em{$\slack$-near feasible}} if for every portal $\rho\in \Portals$ with $\req(\rho)\neq +\infty$
there exists a facility $f\in R$ with $\dist(\rho,f)\leq \req(\rho)+\slack$. That is, we relax all requests by an additive factor of $\slack$.

Finally, for $\gamma\in \R_{\geq 0}$, a solution $R\subseteq \lfac$ is {\em{$\gamma$-close}} in $K$ if
\begin{eqnarray*}
\conncost_K(c,R)\leq \gamma & \qquad & \textrm{for every }c\in \lclients; \textrm{and}\\
\dist(\pi,R)\leq \gamma & \qquad & \textrm{for every }\pi\in \Portals\textrm{ with }\req(\pi)\neq +\infty.
\end{eqnarray*}

The cost of a solution $R$ is defined as
$$\cost(R;K)=\opencost(R)+\sum_{c\in \lclients} \conncost_K(c,R).$$
Note that the connection costs of portals {\em{do not}} contribute to the cost of the solution. They are only used to define (near) feasibility of a solution. 
Thus, every portal essentially puts a hard constraint that there needs to be a facility opened within some distance from it.
By $\OPT(K)$ we denote the minimum cost of a feasible solution to $K$.

The intuitive meaning of predictions and requests in the dynamic programming are as follows. 
In the following, think of dynamic programming over the decomposition provided by Lemma~\ref{lem:tree-decomp} as a recursive algorithm that breaks the given instance into simpler ones (whose number is at most $7$), 
solves them using subcalls, and assembles the obtained solutions into a solution to the input instance.
Whenever we break the instance using some separator, which constists of a constant number of shortest paths, 
we put portals along them using Lemma~\ref{lem:portalization} in all the obtained subinstances.
For every portal $\pi$ we guess in which subinstance lies the closest facility $f$ that is open in the (unknown) optimum solution, 
and we approximately guess the distance $d$ from $\pi$ to this facility (up to additive accuracy $\delta$, to be defined later).
This allows us to define the requests and predictions in subinstances: in the subinstance that is guessed to contain $f$ we put a request $d$ on $\pi$ to make sure that some facility at this distance is indeed 
open there, while in other subinstances we put a prediction $d$ on $\pi$, so that solutions in these subinstances may use a virtual, ``promised'' facility at distance $d$ from $\pi$.

Since recursion has depth $\Oh(\log n)$ by Lemma~\ref{lem:tree-decomp}, condition~\ref{c:depth}, 
the rounding error will accumulate through $\Oh(\log n)$ levels. Therefore, we needed to put $\delta=\Oh(\eps/\log n)$ and make rounding errors of
magnitude $\Oh(\delta)\cdot \OPT$ at each level, so that the total error is kept at $\Oh(\eps)\cdot \OPT$. Precisely, we fix
$$\delta = \frac{\eps}{\log n}.$$

\newcommand{\wNn}{\widetilde{\Nn}}

\paragraph*{Dynamic programming states.} Once we have defined the generalized problem with portals, we may formally define the instances solved in the dynamic programming.
For every vertex $v$ of $H$, we may apply Lemma~\ref{lem:portalization} to $P_v$ and $d=\delta$, thus obtaining a suitable set of vertices 
$\Portals_v\subseteq V(P_v)$ of size at most $\delta^{-1}L+2=\Oh(\eps^{-2}r\log n)$.

For each node $t$ of $T$, we define
$$\lclients_t=\xi^{-1}(L(\bag(t)))\cap \lclients\qquad\textrm{and}\qquad \Portals_t=\bigcup_{uv\in B_t} \Portals_u\cup \Portals_v,$$
where $B_t$ is the set of edges of $H$ dual to the arcs of $\bag(t)$.
Note that by condition~\ref{c:node} of Lemma~\ref{lem:tree-decomp}, we have
$$\Portals_t\subseteq \bigcup_{t'\in \chld(t)}\Portals_{t'}\qquad \textrm{for each non-leaf node }t\textrm{ of }T.$$
Observe also that if $t_0$ is the root of $T$, then $\lclients_{t_0}=\clients_j$ and $\Portals_{t_0}=\emptyset$.
Finally, the following lemma expresses the crucial separation property provided by the decomposition $(T,\bag)$.

\begin{lemma}\label{lem:portal-snap}
Let $s$ and $t$ be nodes of $T$ that are not in the ancestor-descendant relation, and
let $u\in \xi^{-1}(L(\bag(s)))$ and $v\in \xi^{-1}(L(\bag(t)))$. Then there exists a portal $\rho\in \Portals_t$ such that
$$\dist(u,v)\geq \dist(u,\rho)+\dist(\rho,v)-2\delta.$$
Furthermore, the same holds when $s$ is an ancestor of $t$ and $u\in \Portals_s$.
\end{lemma}
\begin{proof}
Let $B$ be the set of edges of $H$ that are dual to the arcs of $\beta(t)$, and let $Z$ be the set of endpoints of these edges.
Consider removing all paths $P_z$ for $z\in Z$ and all edges of $B$ from the plane. Then the plane breaks into several connected components, out of which one consists of exactly the faces of $L(\bag(t))$.
It follows that every path connecting a vertex from $\xi^{-1}(L(\bag(t)))$ with a vertex that does not belong to $\xi^{-1}(L(\bag(t)))$ has to intersect one of the paths $P_z$ for some $z\in Z$.
Observe that $v\in \xi^{-1}(L(\bag(t)))$. Moreover, if $s$ and $t$ are not in the ancestor-descendant relation in $T$, then $L(\bag(s))$ and $L(\bag(t))$ are disjoint, implying $u\notin \xi^{-1}(L(\bag(t)))$.
Also, if $u\in \Portals_s$ and $s$ is an ancestor of $t$, then either $u$ lies on one of the paths $P_z$ for $z\in Z$, or $u\notin \xi^{-1}(L(\bag(t)))$.

In both cases we conclude that the shortest path connecting $u$ and $v$, call it $Q$, has to intersect the path $P_z$ for some $z\in Z$. Let $w$ be any vertex in the intersection of these two paths. 
Then, by Lemma~\ref{lem:portalization}, there exists $\rho\in \Portals_z\subseteq \Portals_t$ such that $\dist(w,\rho)\leq \delta$. We conclude that
\begin{eqnarray*}
\dist(u,v) & =    & \dist(u,w)+\dist(w,v) \\
           & \geq & \dist(u,w)+\dist(w,\rho)+\dist(\rho,w)+\dist(w,v)-2\delta\\
           & \geq & \dist(u,\rho)+\dist(\rho,v)-2\delta,
\end{eqnarray*}
as required.
\end{proof}

For every node $t$ of $T$, we define $\wNn_t$ to be the set of all functions from $\Portals_t$ to $\R\cup \{+\infty\}$.
Further, let $\Nn_t\subseteq \wNn_t$ be the subset of all those functions from $\wNn_t$ that are $(\delta,-5\eps,3r+5\eps,\delta)$-normal;
in the sequel, when saying just {\em{normal}} we mean being $(\delta,-5\eps,3r+5\eps,\delta)$-normal.
While $\wNn_t$ is infinite, $\Nn_t$ is finite and actually of polynomial size.

\begin{lemma}\label{lem:enumerate-Nnt}
For each node $t$ of $T$ we have that $|\Nn_t|\leq n^{\Oh(\eps^{-2}r)}$ and $\Nn_t$ can be enumerated in time $n^{\Oh(\eps^{-2}r)}$.
\end{lemma}
\begin{proof}
By Lemma~\ref{lem:piercing-ring}, for each vertex $u$ of $H$ the number of normal functions on $\Portals_u$ is at most
$(\delta^{-1}r)^2\cdot 2^{\Oh(\delta^{-1}L)}=n^{\Oh(\eps^{-2}r)}$.
Observe that $\Portals_t$ is the union of at most $6$ sets of the form $\Portals_u$, for vertices $u$ that are endpoints of edges dual to the arcs $\bag(t)$.
Hence every normal function on $\Portals_t$ can be described by 
a $6$-tuple of such functions on sets of the form $\Portals_u$ for $u$ as above.
Thus, we have $|\Nn_t|\leq n^{\Oh(\eps^{-2}r)}$ as well. 
Moreover,
since normal functions on $\Portals_u$ can be enumerated in time $n^{\Oh(\eps^{-2}r)}$ for each vertex $u$,
to enumerate $\Nn_t$ it suffices to enumerate all $6$-tuples of functions as above, and filter out those $6$-tuples whose union is either ill-defined or is not Lipschitz with slack $\delta$.
This takes time $n^{\Oh(\eps^{-2}r)}$.
\end{proof}



Now, for every $t\in V(T)$ and pair $\eta=(\pred,\req)\in \wNn_t\times \wNn_t$, we define the instance $K_t(\eta)$ of the generalized problem as follows:
$$K_t(\eta)=(\lclients_t,\Portals_t,\pred,\req).$$
Before the explaining how these instances are going to be solved using dynamic programming, 
let us verify that the subproblem at the root of $T$ corresponds to the instance $J_j$ that we are trying to (approximately) solve.

\begin{lemma}\label{lem:root}
Suppose $t_0$ is the root of $T$ and, noting that $\Portals_{t_0}=\emptyset$, we let $K=K_{t_0}((\emptyset,\emptyset))$.
Then, for any $\slack\geq 0$, every $\slack$-near feasible solution $R$ to $K$ satisfies 
$$\cost(R;J_j)=\cost(R;K).$$
In particular, we have
$$\OPT(J_j)=\OPT(K).$$
\end{lemma}
\begin{proof}
The first assertion follows immediately by observing that the formulas for $\cost(R;J_j)$ and $\cost(R;K)$ are the same, because there are no portals in $K$.
The second assertion follows immediately from the first by observing that every solution $R$ to $K$ is $\slack$-near feasible for any $\slack\geq 0$, because in $K$ there are no portals.
\end{proof}

\paragraph*{Computing transitions.} We first show that the subproblems in the leaves of $T$ can be solved in polynomial time. For this, we use the following lemma.

\newcommand{\DP}{\mathsf{dp}}

\begin{lemma}\label{lem:subset-dp}
There is an algorithm that given an instance $K=(\lclients,\Portals,\pred,\req)$ of the generalized problem and $\slack\geq 0$, 
finds the least expensive $\slack$-near feasible solution to $K$ in time $3^{|\Portals|+k}\cdot n^{\Oh(1)}$, where $k$ is the total number of distinct vertices
on which the clients of $\lclients$ are placed.
\end{lemma}
\begin{proof}
Let $W$ be the set of distinct vertices on which $\lclients$ are placed, and for $u\in W$ let $\gamma(u)$ be the number of clients placed at vertex $u$.
We perform standard dynamic programming over subsets of $\Portals$ and of $W$, where we keep track of the cost of connecting any subset of portals and any subset of vertices of $W$, while introducing candidate
facilities one by one. Precisely, let $f_1,\ldots,f_p$ be the facilities of $\lfac$, enumerated in any order. Then for every $i\in \{0,1,\ldots,p\}$, $A\subseteq \Portals$, and $B\subseteq W$, define value
$\DP[i,A,B]$ to be the smallest cost of a $\slack$-near feasible solution contained in $\{f_1,f_2,\ldots,f_i\}$, where in the near-feasibility check we consider only requests of portals from $A$, and in the connection cost computation we consider only clients placed at vertices from $B$. Then it is easy to see that the function $\DP[\cdot,\cdot,\cdot]$ satisfies the following recursive formula.
\begin{eqnarray*}
\DP[0,A,B] & = & \begin{cases}0 & \textrm{if }A=B=\emptyset,\\ +\infty & \textrm{otherwise;}\end{cases} \\[0.3cm]
\DP[i,A,B] & = & \min(\quad \DP[i-1,A,B],\\
 & & \opencost(f_i)+\min_{\substack{A'\subseteq A,\, B'\subseteq B\colon\\ \forall \pi\in A\setminus A'\ \dist(\pi,f_i)\leq \req(\pi)+\slack}} \DP[i-1,A',B']+\sum_{u\in B\setminus B'} \gamma(u)\cdot \dist(u,f_i)\quad ).
\end{eqnarray*}
Using the above formula, we can in time $3^{|\Portals|+k}\cdot n^{\Oh(1)}$ compute all the $2^{|\Portals|+k}\cdot (p+1)$ values of the function $\DP[\cdot,\cdot,\cdot]$, 
and return $\DP[p,\Portals,W]$ as the sought minimum cost. A $\slack$-near feasible solution attaining this cost can be retrieved from dynamic programming tables by standard means within the same running time.
\end{proof}

\begin{corollary}\label{cor:leaf-dp}
Suppose $t$ is a leaf of $T$ and $\slack\geq 0$ is a given real. 
Then, in total time $n^{\Oh(\eps^{-2}r)}$ one can compute, for each $\eta\in \Nn_t\times \Nn_t$, the least expensive $\slack$-near feasible solution 
$R_{t,\eta}\subseteq \lfac$ to $K_t(\eta)$.
\end{corollary}
\begin{proof}
To compute each solution $R_{t,\eta}$, we apply the algorithm of Lemma~\ref{lem:subset-dp} to instance $K_t(\eta)$ for $\eta\in \Nn_t\times \Nn_t$ and $\slack$.
Since $t$ is a leaf of $T$, all clients in $K_t(\eta)$ lie on the unique face of $L(\bag(t))$ (Lemma~\ref{lem:tree-decomp}, condition~\ref{c:leaf}), hence they are all place on distinct three vertices. 
Therefore, the running time used by each application of the algorithm of Lemma~\ref{lem:subset-dp} is $3^{|\Portals_t|+3}\cdot n^{\Oh(1)}=n^{\Oh(\eps^{-2}r)}$.
Since the number of pairs $\eta\in \Nn_t\times \Nn_t$ is $|\Nn_t|^2\leq n^{\Oh(\eps^{-2}r)}$, the total running time follows.
\end{proof}

\newcommand{\restrict}{\mathsf{restrict}}
\newcommand{\wMm}{\widetilde{\Mm}}
\newcommand{\wUu}{\widetilde{\Uu}}
\newcommand{\wWw}{\widetilde{\Ww}}

We now proceed to the main point: how to compute values for a node of $T$ based on values for its children. We first introduce even more helpful notation.
For a non-leaf node $t$ of $T$, let $\Omega_t=\bigcup_{t'\in \chld(t)}\Portals_t$; then $\Portals_t\subseteq \Omega_t$.


For a non-leaf node $t$ of $T$, define 
$$\wMm_t=\prod_{t'\in \chld(t)} \wNn_t.$$
For each $t'\in \chld(t)$ we have a natural restriction operator $\restrict_{t,t'}\colon \wMm_t\to \wNn_{t'}$ that maps every tuple from $\wMm_t$ to its $t'$-component.
Next, define
\begin{equation*}
\wUu_t = \wNn_t\times \wNn_t\qquad\textrm{and}\qquad \wWw_t = \wMm_t\times \wMm_t.
\end{equation*}
Operator $\restrict_{t,t'}(\cdot)$ can be then regarded as an operator from $\wWw_t$ to $\wUu_{t'}$ by considering acting coordinate-wise.

Having defined sets $\wMm_t$, $\wUu_t$, and $\wWw_t$, we define sets $\Mm_t$, $\Uu_t$, and $\Ww_t$ by replacing $\wNn_t$ with $\Nn_t$ in the definitions.
Since every node of $T$ has at most $7$ children (Lemma~\ref{lem:tree-decomp}, condition~\ref{c:node}), by Lemma~\ref{lem:enumerate-Nnt}
we have that $|\Mm_t|\leq n^{\Oh(\eps^{-2}r)}$ and all sets $\Mm_t$ can be computed in time $n^{\Oh(\eps^{-2}r)}$.
Then we also have that
$$|\Uu_t|,|\Ww_t|\leq n^{\Oh(\eps^{-2}r)}\qquad\textrm{for each node }t\textrm{ of }T,$$
and all the sets $\Uu_t,\Ww_t$ can be computed in time $n^{\Oh(\eps^{-2}r)}$.

We now describe tuples from $\wWw_t$ that may be used in the dynamic programming to combine solutions from smaller subproblems into a solution to a larger subproblem.
The intuition here is that when breaking a subproblem into smaller ones, we have to ensure that requests and predictions appropriately match so that solutions to smaller subproblems can be combined to a solution
to the original subproblem.


\begin{definition}
Consider a non-leaf node $t$ of $T$.
We shall say that a pair $\eta=(\req,\pred)\in \wUu_t$ and a pair $\phi=((\req_{t'})_{t'\in\chld(t)},(\pred_{t'})_{t'\in\chld(t)})\in \wWw_t$ are {\em{compatible}} (denoted $\eta\sim\phi$) 
if the following two conditions hold:
\begin{enumerate}[label=(C\arabic*),ref=(C\arabic*)]
\item\label{cnd:reqs}  For every $\pi\in \Portals_t$ with $\req(\pi)\neq +\infty$ there exists $t'\in \chld(t)$ and $\rho\in \Portals_{t'}$ such that $\req_{t'}(\rho)+\dist(\pi,\rho)\leq \req(\pi)$.
\item\label{cnd:preds} For every $t'\in \chld(t)$ and $\rho\in \Portals_{t'}$ with $\pred_{t'}(\rho)\neq +\infty$, there either exists $\pi\in \Portals_t$ with $\pred(\pi)+\dist(\pi,\rho)\leq \pred_{t'}(\rho)$, or 
there exists $t''\in \chld(t)$ and $\rho'\in \Portals_{t''}$ with $\req_{t''}(\rho')+\dist(\rho',\rho)\leq \pred_{t'}(\rho)$.
\end{enumerate}
\end{definition}

Observe that given $\eta\in \wUu_t$ and $\phi\in \wWw_t$, it can be verified in polynomial time whether $\eta\sim \phi$.


Finally, we formulate and prove two lemmas that will imply the correctness of our dynamic programming.
The first one concerns combining solutions to smaller subproblems into solutions to larger subproblems.
The second one concerns projecting solutions to larger subproblems to solutions to smaller subproblems.

\begin{lemma}\label{lem:combine}
Suppose $t$ is a non-leaf node of $T$ and let $\eta\in \wUu_t$ and $\phi\in \wWw_t$ be compatible.
Suppose further that, for all $t'\in \chld(t)$, $R_{t',\eta_{t'}}$ is a feasible solution to the instance $K_{t'}(\eta_{t'})$, where $\eta_{t'}=\restrict_{t,t'}(\phi)$.
Then
$$R=\bigcup_{t'\in \chld(t)} R_{t',\eta_{t'}}$$
is a feasible solution to the instance $K_t(\eta)$ and, moreover,
$$\cost(R;K_t(\eta))\leq \sum_{t'\in\chld(t)} \cost(R_{t',\eta_{t'}};K_{t'}(\eta_{t'})).$$
\end{lemma}
\begin{proof}
For brevity, we shall denote $R_{t'}=R_{t',\eta_{t'}}$ and $K_{t'}=K_{t'}(\eta_{t'})$. Also, let $\eta=(\pred,\req)$ and $K_t=K_t(\eta)$.

We first verify that $R$ is a feasible solution to $K_t$.
Take any portal $\pi\in \Portals_t$ with $\req(\pi)\neq +\infty$.
Since $\eta\sim\phi$, by~\ref{cnd:reqs} there exists $t'\in \chld(t)$ and $\rho\in \Portals_{t'}$ such that $\req_{t'}(\rho)+\dist(\pi,\rho)\leq \req(\pi)$.
As $R_{t'}$ is a feasible solution to $K_{t'}$, there exists $f\in R_{t'}$ such that $\dist(\rho,f)\leq \req_{t'}(\rho)$.
Then $f\in R$ as well and
$$\dist(\pi,f)\leq \dist(\pi,\rho)+\dist(\rho,f)\leq \dist(\pi,\rho)+\req_{t'}(\rho)\leq \req(\pi),$$
which certifies that the request of $\pi$ is satisfied by $R$.
Hence, $R$ is indeed a feasible solution to $K_t$.

We are left with proving the postulated upper bound on $\cost(R;K_t)$.
Take any client $c\in \lclients_t$.
As $(\lclients_{t'})_{t'\in \chld(t)}$ form a partition of $\lclients_t$, there exists a unique node $t'\in \chld(t)$ satisfying $c\in \lclients_{t'}$.
Then there either exists a facility $f\in R_{t'}$ satisfying
$$\dist(c,f)=\conncost_{K_{t'}}(c;R_{t'})$$
or there exists a portal $\rho\in \Portals_{t'}$ satisfying
$$\dist(c,\rho)+\pred_{t'}(\rho)=\conncost_{K_{t'}}(c;R_{t'}).$$

In the former case, since $R_{t'}\subseteq R$ we can conclude that 
\begin{equation}\label{eq:single-bound}
\conncost_K(c;R)\leq \conncost_{K_{t'}}(c;R_{t'}).
\end{equation}

In the latter case, by~\ref{cnd:preds} either exists $\pi\in \Portals_t$ with $\pred(\pi)+\dist(\pi,\rho)\leq \pred_{t'}(\rho)$, or 
there exists $t''\in \chld(t)$ and $\rho'\in \Portals_{t''}$ with $\req_{t''}(\rho')+\dist(\rho',\rho)\leq \pred_{t'}(\rho)$.
In the first subcase we conclude that
\begin{eqnarray*}
\conncost_K(c;R) & \leq & \dist(c,\pi)+\pred(\pi)\\
                 & \leq & \dist(c,\rho)+\dist(\pi,\rho)+\pred(\pi)\\
                 & \leq & \dist(c,\rho)+\pred_{t'}(\rho)=\conncost_{K_{t'}}(c;R_{t'}),
\end{eqnarray*}
which again establish inequality~\eqref{eq:single-bound} in this subcase. 
On the other hand, in the second subcase there exists a facility $f\in R_{t''}$ with $\dist(\rho',f)\leq \req_{t''}(\rho')$.
As $f\in R$ as well, we infer that
\begin{eqnarray*}
\conncost_K(c;R) & \leq & \dist(c,f)\\
                 & \leq & \dist(c,\rho)+\dist(\rho,\rho')+\dist(\rho',f)\\
                 & \leq & \dist(c,\rho)+\dist(\rho,\rho')+\req_{t''}(\rho')\\
                 & \leq & \dist(c,\rho)+\pred_{t'}(\rho)=\conncost_{K_{t'}}(c;R_{t'}).
\end{eqnarray*}
Hence, again inequality~\eqref{eq:single-bound} is satisfied.

We conclude that in every case, inequality~\eqref{eq:single-bound} holds.
Summing this inequality through all clients $c\in \lclients_t$ and adding $\opencost(R)$ to both sides yields yields that 
$\cost(R;K_t)\leq \sum_{t'\in\chld(t)} \cost(R_{t'};K_{t'})$, as required.
\end{proof}

\begin{lemma}\label{lem:split}
Suppose $t$ is a non-leaf node of $T$.  
Suppose further that $\eta\in \wUu_t$ is such that all predictions involved in $\eta$ are nonnegative,
and $R$ is a $\slack$-near feasible $\gamma$-close solution to $K_t(\eta)$, for some reals $\slack,\gamma>0$.
Then there exist $\phi\in \wWw_t$ that is compatible with $\eta$ and 
 $(\slack+5\delta)$-near feasible $(\gamma+5\delta)$-close solutions $R_{t',\eta_{t'}}\subseteq R$ to instances $K_{t'}(\eta_{t'})$ for $t'\in \chld(t)$, where $\eta_{t'}=\restrict_{t,t'}(\phi)$,
such that
$$\cost(R;K_t(\eta))\geq \sum_{t'\in\chld(t)} \cost(R_{t',\eta_{t'}};K_{t'}(\eta_{t'}))-5\delta|\lclients_t|.$$
Moreover, all request and prediction functions involved in $\phi$ are $(\delta,-\slack-5\delta,\gamma+4\delta,\delta)$-normal, and all predictions involved in $\phi$ are nonnegative.
\end{lemma}
\begin{proof}
Denote $K_t=K_t(\eta)$ and $\eta=(\pred,\req)$. 
For each $t'\in \chld(t)$, let
$$R_{t'}=\xi^{-1}(L(\beta(t')))\cap R.$$
Then $(R_{t'})_{t'\in \chld(t)}$ form a partition of $R$.

For any $t'\in \chld(t)$ and $\rho\in \Portals_{t'}$, we shall say that $\rho$ is {\em{facility-important}} if
\begin{itemize}
\item there exists a facility $f\in R_{t'}$ and a client $c\in \lclients$ served by $f$ in $R$ such that $\dist(c,\rho)+\dist(\rho,f)\leq \dist(c,f)+4\delta$; or
\item there exists a facility $f\in R_{t'}$ and portal $\pi\in \Portals_t$ with $\req(\pi)\neq +\infty$ served by $f$ in $R$ such that $\dist(\pi,\rho)+\dist(\rho,f)\leq \dist(\pi,f)+2\delta$.
\end{itemize}
Further, $\rho$ is {\em{client-important}} if
\begin{itemize}
\item there exists a client $c\in \lclients_{t'}$ and a facility $f\in R$ that serves $c$ in $R$ such that $\dist(c,\rho)+\dist(\rho,f)\leq \dist(c,f)+2\delta$; or
\item there exists a client $c\in \lclients_{t'}$ and a portal $\pi\in \Portals_t$ that serves $c$ in $R$ such that $\dist(c,\rho)+\dist(\rho,\pi)\leq \dist(c,\pi)+2\delta$.
\end{itemize}
We observe the following.

\begin{claim}\label{cl:important-close}
Let $\rho\in \Portals_{t'}$ for some $t'\in \chld(t)$.
If $\rho$ is facility-important, then $$\min_{f\in R_{t'}} \dist(\rho,f) \leq \gamma+4\delta.$$
If $\rho$ is client-important, then $$\min(\min_{f\in R} \dist(\rho,f),\min_{\pi\in \Portals_{t}} \dist(\rho,\pi)+\pred(\pi))\leq \gamma+2\delta$$
\end{claim}
\begin{clproof}
Recall that $R$ is $\gamma$-close in $K_t$.
When $\rho$ is facility-important due to the first alternative in the definition, we have
$$\dist(\rho,f)\leq \dist(c,f)+4\delta\leq \gamma+4\delta;$$
here and in the following, we assume notation from the definition. Also, when $\rho$ is facility-important due to the second alternative, we have
$$\dist(\rho,f)\leq \dist(\pi,f)+2\delta\leq \gamma+2\delta.$$
Now, if $\rho$ is client-important due to the first alternative in the definition, then we have
$$\dist(\rho,f)\leq \dist(c,f)+2\delta\leq \gamma+2\delta.$$
Also, when $\rho$ is facility-important due to the second alternative, we have
$$\dist(\rho,\pi)+\pred(\pi)\leq \dist(c,\pi)+\pred(\pi)+2\delta\leq \gamma+2\delta.$$
This concludes the proof.
\end{clproof}

\newcommand{\rup}{\mathsf{round}^{\uparrow}}
\newcommand{\rdown}{\mathsf{round}^{\downarrow}}

For a real $x$, let $\rdown(x)$ be the largest integer multiple of $\delta$ that is not larger than $x$, and $\rup(x)$ be the smallest integer multiple of $\delta$ that not smaller than $x$. 
That is, 
$$\rdown(x)=\delta\cdot \lfloor x/\delta\rfloor \qquad\textrm{and}\qquad \rup(x)=\delta\cdot \lceil x/\delta\rceil.$$

We now define $\phi=(\pred_{t'},\req_{t'})_{t'\in \chld(t)}$.
Consider any $t'\in \chld(t)$ and $\rho\in \Portals_{t'}$.
We put
\begin{eqnarray*}
 \req_{t'}(\rho) & = &  \begin{cases} +\infty & \textrm{if $\rho$ is not facility-important;}\\
                                      -4\delta+\rdown\left(\min_{f\in R_{t'}} \dist(\rho,f)-\slack\right) & \textrm{otherwise.}\end{cases}\\[0.3cm]
\pred_{t'}(\rho) & = &  \begin{cases} +\infty & \textrm{if $\rho$ is not client-important;}\\
                                      2\delta+\rup\left(\min\left(\min_{f\in R} \dist(\rho,f),\min_{\pi\in \Portals_t} \dist(\rho,\pi)+\pred(\pi)\right)\right) & \textrm{otherwise.}\end{cases}
\end{eqnarray*}
Clearly, functions $\req_{t'}(\cdot)$ and $\pred_{t'}(\cdot)$ are $\delta$-discrete and, as functions of $\rho$ under rounding are Lipschitz, they are also Lipschitz with slack $\delta$.
We are left with verifying that these functions are also $[-\slack-5\delta,\gamma+4\delta]$-bounded, 
$\eta$ and $\phi$ are compatible, $R_{t'}$ is a $(\slack+5\delta)$-near feasible $(\gamma+5\delta)$-close solution to $K_{t'}$ for each $t'\in \chld(t)$, 
where $K_{t'}=K_{t'}(\eta_{t'})$, and that the postulated lower bound on $\cost(R;K_t)$ holds. We prove these properties in the following claims.

\begin{claim}\label{cl:bounded}
For each $t'\in \chld(t)$, the function $\req_{t'}(\cdot)$ is $[-\slack-5\delta,\gamma]$-bounded and the function $\pred_{t'}(\cdot)$ is $[0,\gamma+4\delta]$-bounded.
\end{claim}
\begin{clproof}
First, take any $\rho\in \Portals_{t'}$ that is facility-important (as otherwise $\req_{t'}(\rho)=+\infty$ anyway).
Then $\req_{t'}(\rho)\geq -\slack-5\delta$ by definition and $\req_{t'}(\rho)\leq \gamma$ by Claim~\ref{cl:important-close}.
Next, take any $\rho\in \Portals_{t'}$ that is client-important (as otherwise $\pred_{t'}(\rho)=+\infty$ anyway).
Then $\pred_{t'}(\rho)\geq 2\delta$ by definition and $\pred_{t'}(\rho)\leq \gamma+4\delta$ by Claim~\ref{cl:important-close}.
\end{clproof}

\begin{claim}\label{cl:compatible}
It holds that $\eta$ and $\phi$ are compatible.
\end{claim}
\begin{clproof}
We first verify condition~\ref{cnd:reqs}. Take any $\pi\in \Portals_t$ with $\req(\pi)\neq +\infty$.
Since $R$ is a $\slack$-near feasible solution to instance $K_t$, there exists $f\in R$ such that 
$$\dist(\pi,f)\leq \req(\pi)+\slack.$$
Then $f\in R_{t'}$ for some $t'\in \chld(t)$, and in particular $\xi(f)\in L(\beta(t'))$.
By Lemma~\ref{lem:portal-snap}, there exists a portal $\rho\in \Portals_{t'}$ such that 
\begin{equation}\label{eq:}
\dist(\pi,f)\geq \dist(\pi,\rho)+\dist(\rho,f)-2\delta.
\end{equation}
In particular $\rho$ is facility-important, so 
combining the above with the definition of $\req_{t'}(\rho)$ we obtain
$$\req_{t'}(\rho)\leq \dist(\rho,f)-\slack-4\delta\leq \dist(\pi,f)-\dist(\pi,\rho)+2\delta-\slack-4\delta\leq \req(\pi)-\dist(\pi,\rho)-2\delta;$$
this directly implies~\ref{cnd:reqs}.

We now verify condition~\ref{cnd:preds}. Take any $\rho\in \Portals_{t'}$ for any $t'\in \chld(t)$ with $\pred_{t'}(\rho)\neq +\infty$.
Then $\rho$ is client-important, so there exists a client $c\in \lclients_{t'}$ and either a facility $f\in R$ serving $c$ and satisfying $\dist(c,\rho)+\dist(\rho,f)\leq \dist(c,f)+2\delta$, or a portal
$\pi\in \Portals_t$ serving $c$ such that $\dist(c,\rho)+\dist(\rho,\pi)\leq \dist(c,\pi)+2\delta$.
We consider these two cases separately.

Suppose the first case holds.
Since $f$ serves $c$ in $R$, for any $\pi'\in \Portals_t$ and $f'\in R$, we have
$$\dist(c,f)\leq \dist(c,\pi')+\pred(\pi')\qquad\textrm{and}\qquad\dist(c,f)\leq \dist(c,f').$$
Then we also have
\begin{eqnarray*}
\dist(\rho,f) & \leq & \dist(c,f)-\dist(c,\rho)+2\delta \\ &\leq & \dist(c,\pi')+\pred(\pi')-\dist(c,\rho)+2\delta \\ & \leq & \dist(\rho,\pi')+\pred(\pi')+2\delta,
\end{eqnarray*}
and similarly
\begin{eqnarray*}
\dist(\rho,f) & \leq & \dist(c,f)-\dist(c,\rho)+2\delta\\
& \leq & \dist(c,f')-\dist(c,\rho)+2\delta\\
& \leq & \dist(\rho,f')+2\delta.
\end{eqnarray*}
Therefore, by the definition of $\pred_{t'}(\rho)$, we have
$$\pred_{t'}(\rho)\geq \dist(\rho,f).$$
As $f\in R$, there exists $t''\in \chld(t)$ such that $f\in R_{t''}$.
Then, by Lemma~\ref{lem:portal-snap}, there is a portal $\rho'\in \Portals_{t''}$ such that 
$$\dist(\rho,f)\geq \dist(\rho,\rho')+\dist(\rho',f)-2\delta.$$
We note that
\begin{eqnarray*}
\dist(c,f) & \geq & \dist(c,\rho)+\dist(\rho,f)-2\delta\\
           & \geq & \dist(c,\rho)+\dist(\rho,\rho')+\dist(\rho',f)-4\delta\\
           & \geq & \dist(c,\rho')+\dist(\rho',f)-4\delta,
\end{eqnarray*}
implying that $\rho'$ is facility-important. Therefore, by the definition of $\req_{t''}(\rho')$ we infer that
$$\req_{t''}(\rho')\leq \dist(\rho',f)-\slack-4\delta\leq \dist(\rho',f)-4\delta.$$
Combining all the above we infer that
$$\pred_{t'}(\rho)\geq \dist(\rho,f)\geq \dist(\rho,\rho')+\dist(\rho',f)-4\delta\geq \dist(\rho,\rho')+\req_{t''}(\rho'),$$
which establishes~\ref{cnd:preds} in this case.

Suppose now the second case holds.
Since $\pi$ serves $c$ in $R$, for any $\pi'\in \Portals_t$ and $f'\in R$, we have
$$\dist(c,\pi)+\pred(\pi)\leq \dist(c,\pi')+\pred(\pi')\qquad\textrm{and}\qquad\dist(c,\pi)+\pred(\pi)\leq \dist(c,f').$$
Using the same reasoning as in the first case, but considering expression $\dist(c,\pi)+\pred(\pi)$ instead of $\dist(c,f)$, we infer that
$$\pred_{t'}(\rho)\geq \dist(\rho,\pi)+\pred(\pi),$$
which establishes~\ref{cnd:preds} in this case as well.
\end{clproof}

For the next claim, recall that $(\lclients_{t'})_{t'\in \chld(t)}$ form a partition of $\lclients_t$.
\begin{claim}\label{cl:bound-c}
Let $c \in \lclients_t$ and let $t'\in \chld(t)$ be the unique node satisfying $c\in \lclients_{t'}$. Then the following holds.
\begin{equation}\label{eq:single-apx}
\conncost_{K_{t'}}(c,R_{t'})\leq \conncost_{K_t}(c,R)+5\delta.
\end{equation}
\end{claim}
\begin{clproof}
By the definition of $\conncost_{K_t}(c,R)$, there either exists a portal $\pi\in \Portals_t$ such that
$$\conncost_{K_{t}}(c,R)=\dist(c,\pi)+\pred(\pi),$$
or there exists a facility $f\in R$ such that
$$\conncost_{K_{t}}(c,R)=\dist(c,f).$$

Suppose the first case holds.
By Lemma~\ref{lem:portal-snap}, there exists a portal $\rho\in \Portals_{t'}$ such that 
$$\dist(c,\pi)\geq \dist(c,\rho)+\dist(\rho,\pi)-2\delta.$$
In particular, $\rho$ is facility-important. By the definition of $\pred_{t'}(\rho)$, we have
$$\pred_{t'}(\rho)\leq \dist(\rho,\pi)+\pred(\pi)+3\delta.$$
By combining the above we conclude that
\begin{eqnarray*}
\conncost_{K_{t'}}(c,R_{t'}) & \leq & \dist(c,\rho)+\pred_{t'}(\rho)\\
                             & \leq & \dist(c,\rho)+\dist(\rho,\pi)+\pred(\pi)+3\delta\\
                             & \leq & \dist(c,\pi)+\pred(\pi)+5\delta=\conncost_{K_t}(c,R)+5\delta;
\end{eqnarray*}
This establishes~\eqref{eq:single-apx} in this case.

Now suppose the second case holds.
Since $(R_{t'})_{t'\in \chld(t)}$ is a partition of $R$, there exists $t''\in \chld(t)$ such that $f\in R_{t''}$.
If $t''=t'$, then we have
$$\conncost_{K_{t'}}(c,R_{t'})\leq \dist(c,f)=\conncost_{K_t}(c,R),$$
so~\eqref{eq:single-apx} indeed holds in this situation. Assume then that $t''\neq t'$.
By Lemma~\ref{lem:portal-snap}, there exists a portal $\rho\in \Portals_{t'}$ such that
$$\dist(c,f)\geq \dist(c,\rho)+\dist(\rho,f)-2\delta.$$
In particular, $\rho$ is facility-important.
By the definition of $\pred_{t'}(\rho)$, we have
$$\pred_{t'}(\rho)\leq \dist(\rho,f)+3\delta$$
By combining the above we conclude that
\begin{eqnarray*}
\conncost_{K_{t'}}(c,R_{t'}) & \leq & \dist(c,\rho)+\pred_{t'}(\rho)\\
                             & \leq & \dist(c,\rho)+\dist(\rho,f)+3\delta\\
                             & \leq & \dist(c,f)+5\delta = \conncost_{K_t}(c,R)+5\delta.
\end{eqnarray*}
Hence, again~\eqref{eq:single-apx} holds in this case.
\end{clproof}
\begin{claim}\label{cl:bound}
It holds that $\cost(R;K_t)\geq \sum_{t'\in\chld(t)} \cost(R_{t'};K_{t'})-5\delta|\lclients_t|$.
\end{claim}
\begin{clproof}
The claimed upper bound on $\cost(R;K_t)$ follows by adding the thesis of Claim~\ref{cl:bound-c} through all clients $c\in \lclients_t$, and adding the opening costs of facilities of $R$ to both sides.
\end{clproof}

\begin{claim}\label{cl:feasible}
For each $t'\in \chld(t)$, $R_{t'}$ is a $(\slack+5\delta)$-near feasible $(\gamma+5\delta)$-close solution to $K_{t'}$.
\end{claim}
\begin{clproof}
We first verify the $(\slack+5\delta)$-near feasibility.
Take any $\rho\in \Portals_{t'}$ with $\req_{t'}(\rho)\neq +\infty$; then $\rho$ is facility-important.
By the definition of $\req_{t'}(\rho)$, there exists a facility $f\in R_{t'}$ such that 
$$\req_{t'}(\rho)\geq \dist(\rho,f)-\slack-5\delta,\qquad
\textrm{implying}\qquad\dist(\rho,f)\leq \req_{t'}(\rho)+\slack+5\delta,$$
as required.

We now verify the $(\gamma+5\delta)$-closeness. 
Claim~\ref{cl:bound-c} asserts that
for each $c\in \lclients_{t'}$ we have
$$\conncost_{K_{t'}}(c,R_{t'})\leq \conncost_{K_t}(c,R)+5\delta,$$
which by $\gamma$-closeness of $R$ implies that 
$$\conncost_{K_{t'}}(c,R_{t'})\leq \gamma+5\delta.$$
This is the first condition of the $(\gamma+5\delta)$-closeness.
For the second condition, consider any $\rho\in \Portals_{t'}$ with $\req_{t'}(\rho)\neq +\infty$.
In particular, $\rho$ is facility-important, so there exists a facility $f\in R_{t'}$ and either a client $c\in \lclients$ served by $f$ such that $\dist(c,\rho)+\dist(\rho,f)\leq \dist(c,f)+4\delta$,
or a portals $\pi\in \Portals_t$ served by $f$ such that $\dist(\pi,\rho)+\dist(\rho,f)\leq \dist(\pi,f)+2\delta$.
By $\gamma$-closeness of $R$ in $K$, in the first case we have
$$\dist(\rho,f)\leq \dist(c,f)-\dist(c,\rho)+4\delta\leq \gamma+4\delta,$$
while in the second case we have
$$\dist(\rho,f)\leq \dist(\pi,f)-\dist(\pi,\rho)+2\delta\leq \gamma+2\delta.$$
In both cases, we conclude that $\dist(\rho,f)\leq \gamma+5\delta$, as required.
\end{clproof}

Claims~\ref{cl:bounded},~\ref{cl:compatible},~\ref{cl:bound}, and~\ref{cl:feasible} conclude the proof.
\end{proof}

\paragraph*{The algorithm.} We are finally ready to present the whole algorithm. 
First, using the algorithm of Lemma~\ref{lem:tree-decomp} in polynomial time we compute the tree $T$ together with sets $\bag(t)$ for nodes $t$ of $T$.
For each node $t$ we compute the portal set $\Portals_t$ and the set of functions $\Nn_t$, as explained before; this takes total time $n^{\Oh(\eps^{-2}r)}$, since $T$ is of size $n^{\Oh(1)}$.
Sets $\Nn_t$ give rise to sets $\Uu_t$ and $\Ww_t$ as defined before.

The remaining, main part of the algorithm is summarized using pseudo-code as Algorithm~$\AlgName$.
We process the nodes of $T$ in a bottom-up manner. 
For each node $t$, say at depth $i$, and each $\eta\in \Uu_t$, we construct the instance $K_{t}(\eta)$ and compute an $5\eps$-near feasible solution $R_{t,\eta}$ to it as follows.
If $t$ is a leaf, we use the algorithm of Corollary~\ref{cor:leaf-dp} to compute the least expensive $5\eps$-near feasible solution $R_{t,\eta}$.
Otherwise, we iterate over all $\phi\in \Ww_t$ such that $\eta$ and $\phi$ are compatible, and consider all candidate solutions $R(\phi)$ defined as
$$R=\bigcup_{t'\in \chld(t)} R_{t',\restrict_{t,t'}(\phi)}.$$
Here, $R_{t',\restrict_{t,t'}(\phi)}$ is the pre-computed soluton to the instance $K_{t'}(\restrict_{t,t'}(\phi))$. Out of these candidate solutions we take the least expensive one and we declare it as $R_{t,\eta}$.

Finally, we return $R=R_{t_0,(\emptyset,\emptyset)}$ as computed solution, where $t_0$ is the root of $T$. This concludes the description of the algorithm and we are left with analyzing its running time and approximation
guarantee.

\begin{algorithm}[h!]
  
  \KwIn{Instance $J_j$, tree $T$, and sets $\Uu_t,\Ww_t$ for nodes $t$ of $T$} 
  \KwOut{Solution $R$ to $J_j$}  \Indp \BlankLine
  \For{{\bf{each}} node $t$ of $T$ in bottom-up order}{
    \For{{\bf{each}} $\eta\in \Uu_t$}{
      \If{$t$ is a leaf}{
        $R_{t,\eta}\leftarrow$ minimum-cost $5\eps$-near feasible solution to $K_{t,\eta}$, computed using Corollary~\ref{cor:leaf-dp}\\
      }\Else{
        $R_{t,\eta}\leftarrow \bot$\\
        \For{{\bf{each}} $\phi\in \Ww_t$ such that $\eta\sim \phi$}{
           $S\leftarrow \bigcup_{t'\in \chld(t)} R_{t',\restrict_{t,t'}(\phi)}$\\
           \If{$R_{t,\eta}=\bot$ or $\cost(S,K_{t,\eta})<\cost(R_{t,\eta},K_{t,\eta})$}{
              $R_{t,\eta}\leftarrow S$
           }
        }
      }
    }
  }
  $R\leftarrow R_{t_0,(\emptyset,\emptyset)}$, where $t_0$ is the root of $T$\\
  \KwRet{$R$}   
\caption{Algorithm $\AlgName$}
  \label{alg:main-ptas}
\end{algorithm}

\begin{lemma}\label{lem:dp-runtime}
Algorithm~$\AlgName$ runs in time $n^{\Oh(\eps^{-2}r)}$.
\end{lemma}
\begin{proof}
It suffices to observe that, by Corollary~\ref{cor:leaf-dp} and Lemma~\ref{lem:combine}, the time spent on processing every node of $T$ is bounded by $n^{\Oh(\eps^{-2}r)}$.
Since the number of nodes of $T$ is $n^{\Oh(1)}$, the total running time follows.
\end{proof}

\begin{lemma}\label{lem:dp-apx}
Algorithm~$\AlgName$ returns a solution $R$ to the instance $J_j$ satisfying
$$\cost(R;J_j)\leq \OPT(J_j)+10\eps|\clients_j|.$$
\end{lemma}
\begin{proof}
Let $D\subseteq \fac_j$ be an optimum solution to the instance $J_j$.
By Lemma~\ref{lem:root}, $D$ is also an optimum feasible solution to the instance $K=K_{t_0}((\emptyset,\emptyset))$, where $t_0$ is the root of $T$,
Furthermore, by Lemma~\ref{lem:close-openfac-Jj} we infer that $D$ is $3r$-close in $K$.

By applying Lemma~\ref{lem:split} in a top-down manner along the tree $T$,
we obtain, for every node $t$ of $T$, an element $\eta_t\in \wUu_t$ and a solution $D_t$ to the instance $K_t(\eta_t)$ such that the following holds:
\begin{itemize}
\item whenever $t$ is not a leaf, we have that $\phi_t=(\eta_{t'})_{t'\in \chld(t)}$ is compatible with $\eta_t$;
\item $D_t$ is a $(5i\delta)$-near feasible $(3r+5i\delta)$-close solution in $K_t(\eta_t)$, where $i$ is the depth of $t$ in $T$;
\item all request and prediction functions involved in $\eta_t$ are $(\delta,-5i\delta,3r+5i\delta,\delta)$-normal, and all prediction functions are nonnegative;
\item whenever $t$ is not a leaf, it holds that
\begin{equation}\label{eq:split-opt}
\cost(D_t;K_t(\eta_t))\geq \sum_{t'\in \chld(t)} \cost(D_{t'},K_{t'}(\eta_{t'})) - 5\delta|\lclients_t|.
\end{equation}
\end{itemize}
Recall that $T$ has depth at most $\log n$.
Therefore, $5i\delta\leq 5\eps$ whenever $i$ is the depth of a node in $t$, implying that all request and prediction functions involved in elements $\eta_t$ are $(\delta,-5\eps,3r+5\eps,\delta)$-normal.
We infer that
\begin{equation}\label{eq:normalized}
\eta_t\in \Uu_t\qquad\textrm{for each node }t.
\end{equation}

Recall also that for each non-leaf node $t$ of $T$, we have that $\{\lclients_{t'}\colon t'\in \chld(t)\}$ form a partition of $\lclients_t$.
Therefore, by combining inequalities~\eqref{eq:split-opt} in a bottom-up manner along $T$ we infer that
\begin{equation}\label{eq:Ddecomp}
\cost(D;K)\geq \sum_{t\colon \textrm{leaf of }T} \cost(D_t,K_t(\eta_t)) - 5\delta\log n|\clients_j|=\sum_{t\colon \textrm{leaf of }T} \cost(D_t,K_t(\eta_t)) - 5\eps|\clients_j|.
\end{equation}
Again, as $i\delta\leq \eps$ whenever $i\leq \log n$, for each leaf $t$ of $T$ the solution $D_t$ is $5\eps$-near feasible in $K_t(\eta_t)$.
Hence, due to~\eqref{eq:normalized} for each leaf $t$ the algorithm computes an $5\eps$-near feasible solution $R_t$ to $K_t(\eta_t)$ satisfying
\begin{equation}\label{eq:RvsD}
\cost(R_t;K_t(\eta_t))\leq \cost(D_t;K_t(\eta_t)).
\end{equation}

For each non-leaf node $t$ of $T$, define solution $R_t$ to instance $K_t(\eta_t)$ by a bottom-up induction: $R_t=\bigcup_{t'\in \chld(t)} R_{t'}$.
Then by~\eqref{eq:normalized} and the fact that $\eta_t\sim \phi_t$ for every non-leaf $t$, we have that for each node $t$, the algorithm computes a solution to $\eta_t$ of cost at most $\cost(R_t;K_t(\eta_t))$.
In particular, if we denote $R=R_{t_0}$, where $t_0$ is the root of $T$, then the solution returned by the algorithm has cost at most $\cost(R;K)$.
Hence, we proceed with upper bounding $\cost(R;K)$.

For each node $t$ of $T$ let us define tuples of functions $\eta'_t$ and $\phi'_t$ (here, only when $t$ is not a leaf) as follows:
$$\eta'_t=\eta_t+5\eps\qquad\textrm{and}\qquad\phi'_t=\phi_t+5\eps.$$
That is, $\eta'_t$ is obtained from $\eta_t$ by adding $5\eps$ to all requests and all predictions on all portals of $\Portals_t$, and similarly for $\phi_t$.
Note that for each non-leaf node $t$ of $T$, we still have the following properties:
\begin{itemize}
\item $\eta'_{t'}=\restrict_{t,t'}(\phi'_t)$ for each $t'\in \chld(t)$, and
\item $\eta'_t$ and $\phi'_t$ are compatible.
\end{itemize}
However, the $5\eps$ shift in requests and predictions makes the following assertion hold for each leaf $t$ of $T$:
\begin{equation}\label{as:leaf}
R_t\textrm{ is a feasible solution to }K_t(\eta'_t)\textrm{ with }\cost(R_t;K_t(\eta'_t))\leq \cost(R_t;K_t(\eta_t))+5\eps|\lclients_t|.
\end{equation}
That is, we obtained feasibility instead of $5\eps$-near feasibility at the cost of increasing the cost of the solution.

Denoting $\desc(t)$ the set of leaves of $T$ that are descendants of $t$, we
may now apply Lemma~\ref{lem:combine} through a bottom-up induction along the tree $T$ to infer the following for each node $t$ of $T$:
\begin{equation}\label{as:bot-up}
R_t\textrm{ is a feasible solution to }K_t(\eta'_t)\textrm{ with }\cost(R_t;K_t(\eta'_t))\leq \sum_{t'\in \desc(t)}\cost(R_{t'};K_{t'}(\eta'_{t'})).
\end{equation}
In particular, assertion~\eqref{as:bot-up} holds for the root $t_0$ of $T$.
Then, we may use assertions~\eqref{eq:Ddecomp},~\eqref{eq:RvsD}, and~\eqref{as:leaf} to infer the following:
\begin{eqnarray*}
\cost(R;K) & \leq & \sum_{t\colon \textrm{leaf of }T} \cost(R_t;K_t(\eta'_t))\\
           & \leq & \sum_{t\colon \textrm{leaf of }T} \cost(R_t;K_t(\eta_t)) + 5\eps|\clients_j|\\
           & \leq & \sum_{t\colon \textrm{leaf of }T} \cost(D_t;K_t(\eta_t)) + 5\eps|\clients_j|\\
           & \leq & \cost(D;K) + 10\eps|\clients_j|.
\end{eqnarray*}
It now suffices to use Lemma~\ref{lem:root} to infer that $\cost(R;K)=\cost(R;J_j)$ and $\cost(D;K)=\cost(D;J_j)$; this combined with the above concludes the proof.
\end{proof}

We now conclude the proof of Lemma~\ref{lem:same-rad}.
Apply Algorithm~$\AlgName$ to each instance $J_j$ for which $C_j$ is non-empty, yielding a solution $R_j$.
As the number of such instances is at most $n$, by Lemma~\ref{lem:dp-runtime} this takes total time $n^{\Oh(\eps^{-2}r)}$.
As the final solution return $R=S\cup \bigcup_{j\in \N} R_j$, where we set $R_j=\emptyset$ whenever $C_j=\emptyset$.
Then, by Lemmas~\ref{lem:layering-separation} and~\ref{lem:dp-apx} we have
$$\cost(R;J')\leq \eps\cdot M+\sum_{j\in \N}\cost(R_j;J_j)\leq \eps\cdot M+10\eps\cdot |\clients|+\sum_{j\in \N}\OPT(J_j)\leq \OPT(J')+\eps\cdot M+10\eps\cdot |\clients|.$$
Finally, we observe that since $\dist(c,f)\geq 1$ for each client $c\in \cluster(f)$, we have
$$|\clients|\leq \sum_{f\in \hintSol} \sum_{c\in \cluster(f)} \dist(c,f)\leq M.$$
Therefore, we conclude that
$$\cost(R;J')\leq \OPT(J')+11\eps\cdot M.$$
It now remains to apply Corollary~\ref{cor:trim} to infer the same inequality for instance $J$ instead of $J'$, and to rescale $\eps$ by a multiplicative factor of $11$.